\documentclass[journal, draftcls,onecolumn]{IEEEtran}
\usepackage{enumerate}
\usepackage{cite}
\usepackage{amscd}
\usepackage{dsfont}
\usepackage{latexsym}
\usepackage{array}

\usepackage{tikz}\usetikzlibrary{calc}
\usepackage{tabularx}
\usepackage{epsfig}
\usepackage{graphicx,subfigure}
\usepackage{amsmath,mathrsfs}
\usepackage{amsthm}
\usepackage{amssymb}  
\usepackage{amsbsy}
\usepackage{color}
\usepackage{stfloats}
\usepackage{algorithm}
\usepackage{algpseudocode}
\usepackage{bm}
\usepackage{pifont}

\newtheoremstyle{note}
{3pt}
{1pt}
{}
{\parindent}
{\itshape}
{:}
{.5em}
{\thmname{#1}\thmnumber{ #2}\thmnote{\thmnote{ (#3)}}}
\theoremstyle{note}

\newtheorem{lem}{Lemma}
\newtheorem{ther}{Theorem}
\newtheorem{deft}{Definition}
\newtheorem{prop}{Proposition}

\theoremstyle{definition}
\newtheorem{rem}{Remark}
\newtheoremstyle{dotless}{}{}{\itshape}{}{\bfseries}{}{ }{}
\theoremstyle{dotless}

\newcommand{\Z}{\mathsf{Z}}
\newcommand{\M}{\mathsf{M}}
\newcommand{\E}{\mathsf{E}}
\newcommand{\F}{\mathsf{F}}
\newcommand{\V}{\mathbb{V}}

\newcommand{\R}{\mathbb{R}}

\newcommand{\X}{\mathsf{X}}
\newcommand{\U}{\mathsf{U}}
\newcommand{\Y}{\mathsf{Y}}
\newcommand{\D}{\mathbb{D}}

\newcommand{\W}{\mathsf{W}}
\newcommand {\aplt} {\ {\raise-.5ex\hbox{$\buildrel<\over{\mbox{\scriptsize $\sim$}}$}}\ }
\providecommand{\abs}[1]{\ensuremath{\left\lvert #1 \right\rvert}}

\DeclareMathOperator{\SNR}{\mathsf{SNR}}
\newcommand{\argmax}{\operatornamewithlimits{argmax}}

\begin{document}
%
\title{Achieving Secrecy Capacity of the Gaussian Wiretap Channel with Polar Lattices}
\author{Ling~Liu,
        Yanfei~Yan,
	        and Cong Ling~\IEEEmembership{Member,~IEEE}
    \thanks{This work was supported in part by FP7 project PHYLAWS (EU FP7-ICT 317562) and in part by the China Scholarship Council. This work was presented at the IEEE Int. Symp. Inform. Theory (ISIT), Honolulu, USA, 2014 and the IEEE Inform. Theory Workshop, Jerusalem, ISRAEL, 2015.}
	\thanks{Ling Liu, Yanfei Yan and Cong Ling are with the Department of Electrical and Electronic Engineering,
	 Imperial College London, London, UK (e-mails: l.liu12@imperial.ac.uk, y.yan10@imperial.ac.uk, cling@ieee.org).}
}

\maketitle
\begin{abstract}
In this work, an explicit scheme of wiretap coding based on polar lattices is proposed to achieve the secrecy capacity of the additive white Gaussian noise (AWGN) wiretap channel. Firstly, polar lattices are used to construct secrecy-good lattices for the mod-$\Lambda_s$ Gaussian wiretap channel. Then we propose an explicit shaping scheme to remove this mod-$\Lambda_s$ front end and extend polar lattices to the genuine Gaussian wiretap channel. The shaping technique is based on the lattice Gaussian distribution, which leads to a binary asymmetric channel at each level for the multilevel lattice codes. By employing the asymmetric polar coding technique, we construct an AWGN-good lattice and a secrecy-good lattice with optimal shaping simultaneously. As a result, the encoding complexity for the sender and the decoding complexity for the legitimate receiver are both $O(N\log N\log(\log N))$. The proposed scheme is proven to be semantically secure.
\end{abstract}

\IEEEpeerreviewmaketitle

\section{Introduction}
Wyner \cite{wyner1} introduced the wiretap channel model and showed that both reliability and confidentiality could be attained by coding without any key bits if the channel between the sender and the eavesdropper (wiretapper's channel $W$) is degraded with respect to the channel between the sender and the legitimate receiver (main channel $V$). The goal of wiretap coding is to design a coding scheme that makes it possible to communicate both reliably and securely between the sender and the legitimate receiver. Reliability is measured by the decoding error probability for the legitimate user, namely $\lim\limits_{N\rightarrow\infty} \text{Pr}\{\widehat{\M}\neq \M\}=0$, where $N$ is the length of transmitted codeword, $\M$ is the confidential message and $\widehat{\M}$ is its estimate. Secrecy is measured by the mutual information between $\M$ and the signal received by the eavesdropper $\Z^{[N]}$. In this work, we will follow the strong secrecy condition proposed by Csisz\'{a}r \cite{csis1}, i.e., $\lim\limits_{N\rightarrow\infty}I(\M;\Z^{[N]})=0$, which is more widely accepted than the weak secrecy criterion $\lim\limits_{N\rightarrow\infty}\frac{1}{N}I(\M;\Z^{[N]})=0$. In simple terms, the secrecy capacity is defined as the maximum achievable rate under both the reliability and strong secrecy conditions. When $W$ and $V$ are both symmetric, and $W$ is degraded with respect to $V$, the secrecy capacity is given by $C(V)-C(W)$ \cite{wiretapCapacity}, where $C(\cdot)$ denotes the channel capacity.

In the study of strong secrecy, plaintext messages are often assumed to be
random and uniformly distributed. From a cryptographic point of view, it is crucial that the security does not rely on the distribution of the message.
This issue can be resolved by using the standard notion of \emph{semantic security}~\cite{GoMi84} which
means that, asymptotically, it is impossible to estimate any function of
the message better than to guess it without accessing~$\Z^{[N]}$ at all.
The relation between
strong secrecy and semantic security was recently revealed in~\cite{Bellare_Tessaro_Vardy,cong2},
namely, semantic security is equivalent to achieving strong secrecy
for all distributions $p_{\M}$ of the plaintext messages:
\begin{equation} \label{mis}
\lim_{N\to \infty}\max_{p_{\M}} {I}(\M;\Z^{[N]}) = 0.
\end{equation}

\begin{figure}[h]
    \centering
    \includegraphics[width=12cm]{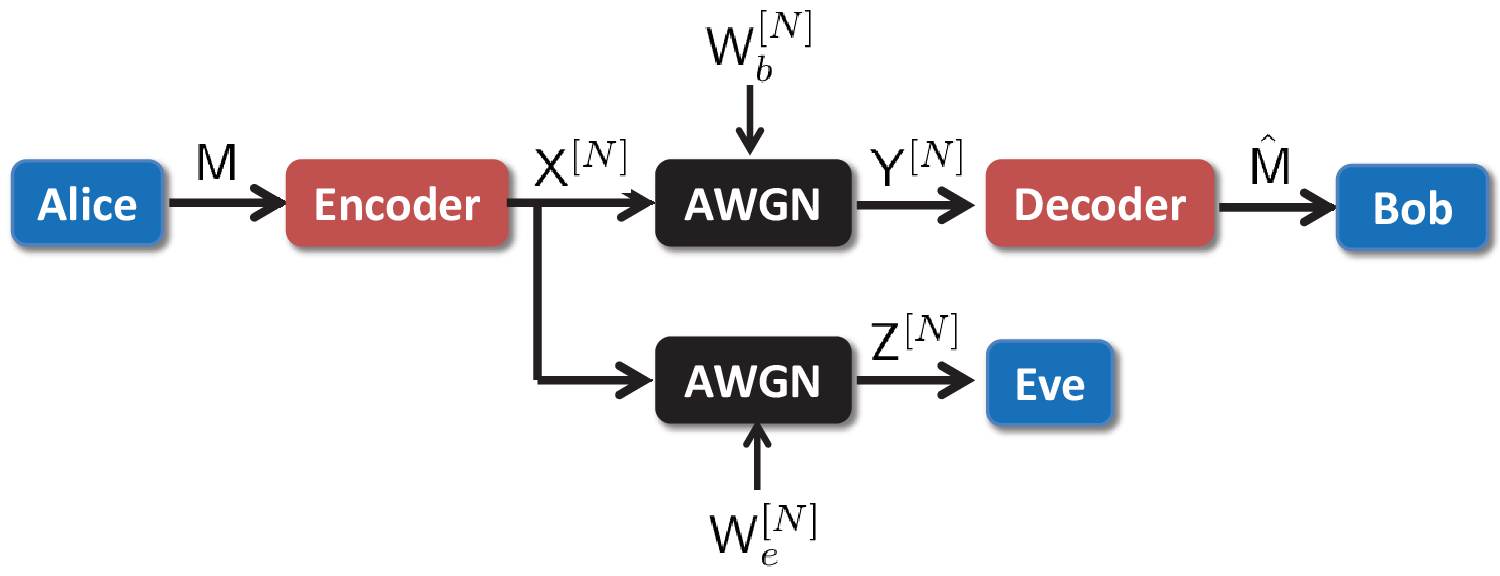}
    \caption{The Gaussian wiretap channel.}
    \label{fig:wiretap}
\end{figure}

In this work, we construct lattice codes for the Gaussian wiretap channel (GWC) which is shown in Fig. \ref{fig:wiretap}. The confidential message $\M$ drawn from the message set $\mathcal{M}$ is encoded by the sender (Alice) into an $N$-dimensional codeword $\X^{[N]}$. The outputs $\Y^{[N]}$ and $\Z^{[N]}$ received by the legitimate receiver (Bob) and the eavesdropper (Eve) are respectively given by
\begin{eqnarray}
\left\{\begin{aligned}
&\Y^{[N]}=\X^{[N]}+\W_{b}^{[N]} \notag\\
&\Z^{[N]}=\X^{[N]}+\W_{e}^{[N]}, \notag\
\end{aligned}\right.
\notag\
\end{eqnarray}
where $\W_{b}^{[N]}$ and $\W_{e}^{[N]}$ are $N$-dimensional Gaussian noise vectors with zero mean and variance $\sigma_{b}^{2}$, $\sigma_{e}^{2}$ respectively. The channel input $\X^{[N]}$ satisfies the power constraint $P_s$, i.e.,
\begin{eqnarray}
\frac{1}{N}E[\|\X^{[N]}\|^2]\leq P_s.
\notag\
\end{eqnarray}

Polar codes \cite{arikan2009channel} have shown their great potential in solving the wiretap coding problem. The polar coding scheme proposed in \cite{polarsecrecy}, combined with the block Markov coding technique \cite{NewPolarSchemeWiretap}, was proved to achieve the strong secrecy capacity when $W$ and $V$ are both binary-input symmetric channels, and $W$ is degraded with respect to $V$. More recently, polar wiretap coding has been extended to general wiretap channels (not necessarily degraded or symmetric) in \cite{BargWiretap} and \cite{GeneralWiretap}. For continuous channels such as the GWC, there also has been notable progress in wiretap lattice coding. On the theoretical aspect, the existence of lattice codes achieving the secrecy capacity to within $\frac{1}{2}$ nat under the strong secrecy as well as semantic security criterion was demonstrated in \cite{cong2}. On the practical aspect, wiretap lattice codes were proposed in \cite{LatticeBelf} and \cite{GaussianFadingWiretap} to maximize the eavesdropper's decoding error probability.

\subsection{Our contribution}
Polar lattices, the counterpart of polar codes in the Euclidean space, have already been proved to be additive white Gaussian noise (AWGN)-good \cite{yan2} and further to achieve the AWGN channel capacity with lattice Gaussian shaping \cite{polarlatticeJ}\footnote{Please refer to \cite{BK:Zamir,AbbeAWGN,BarronAWGN} for other methods of achieving the AWGN channel capacity.}.  Motivated by \cite{polarsecrecy}, we will propose polar lattices to achieve both strong secrecy and reliability over the mod-$\Lambda_s$ GWC. Conceptually, this polar lattice structure can be regarded as a secrecy-good lattice $\Lambda_e$ nested within an AWGN-good lattice $\Lambda_b$ ($\Lambda_e \subset \Lambda_b$). Further, we will propose a Gaussian shaping scheme over $\Lambda_b$ and $\Lambda_e$, using the multilevel asymmetric polar coding technique. As a result, we will accomplish the design of an explicit lattice coding scheme which achieves the secrecy capacity of the GWC with semantic security.
\begin{itemize}
  \item The first technical contribution of this paper is the explicit construction of secrecy-good polar lattices for the mod-$\Lambda_s$ GWC and the proof of their secrecy capacity-achieving. This is an extension of the binary symmetric wiretap coding \cite{polarsecrecy} to the multilevel coding scenario, and can also be considered as the construction of secrecy-good polar lattices for the GWC without the power constraint. The construction for the mod-$\Lambda_s$ GWC provides considerable insight into wiretap coding for the genuine GWC, without deviating to the technicality of Gaussian shaping. This work is also of independent interest to other problems of information theoretic security, e.g., secret key generation from Gaussian sources \cite{CongSecKey}.
  \item Our second contribution is the Gaussian shaping applied to the secrecy-good polar lattice, which follows the technique of \cite{cong2,polarlatticeJ}. The resultant coding scheme is proved to achieve the secrecy capacity of the GWC. It is worth mentioning that our proposed coding scheme is not only a practical implementation of the secure random lattice coding in \cite{cong2}, but also an improvement in the sense that we successfully remove the constant $\frac{1}{2}$-nat gap to the secrecy capacity\footnote{The $\frac{1}{2}$-nat gap in \cite{cong2} was due to a requirement on the volume-to-noise ratio of the secrecy-good lattice. In this paper, we employ mutual information, rather than via the flatness factor, to directly bound information leakage, thereby removing that requirement of the secrecy-good lattice.}. 
\end{itemize}

\subsection{Comparison with the extractor-based approach}
Invertible randomness extractors were introduced into wiretap coding in \cite{HayashiWiretap,InvertExtraWiretap,Bellare_Tessaro_Vardy}. The key idea is that an extractor is used to convert a capacity-achieving code with rate close to $C(V)$ for the main channel into a wiretap code with the rate close to $C(V)-C(W)$. Later, this coding scheme was extended to the GWC in \cite{VardyGWC}. Besides, channel resolvability \cite{BlockResolv} was proposed as a tool for wiretap codes. An interesting connection between the resolvability and the extractor was revealed in \cite{BlochITW}.

The proposed approach and the one based on invertible extractors have their respective advantages. The extractor-based approach is modular, i.e., the error-correction code and extractor are realized separately; it is possible to harness the results of invertible extractors in literature. The advantage of our lattice-based scheme is that the wiretap code designed for Eve is nested within the capacity-achieving code designed for Bob, which represents an integrated approach. More importantly, lattice codes are attractive for emerging applications in network information theory thanks to their useful structures \cite{BK:Zamir}, \cite{ComptForward}; thus the proposed scheme may fit better with this landscape when security is a concern \cite{Liang2009}.

\subsection{Outline of the paper}

The paper is organized as follows: Section II presents some preliminaries of lattice codes. The binary polar codes and multilevel lattice structure \cite{forney6} are briefly reviewed in Section III, where the original polar wiretap coding scheme in \cite{polarsecrecy} is slightly modified to be compatible to the following shaping operation. In Section IV, we construct secrecy-good polar lattices for the mod-$\Lambda_s$ GWC. In Section V, we show how to implement the discrete Gaussian shaping over the polar lattice to remove the mod-$\Lambda_s$ front end, using the polar coding technique for asymmetric channels. Then we prove that our wiretap lattice coding achieves the secrecy capacity with shaping. Finally, we discuss the relationship between the lattice constructions with and without shaping in Section VI.

\subsection{Notation} \color{black}
All random variables (RVs) will be denoted by capital letters. Let $P_\X$ denote the probability distribution of a RV $\X$ taking values $x$ in a set $\mathcal{X}$ and let $H(\X)$ denote its entropy. For multilevel coding, we denote by $\X_\ell$ a RV $\X$ at level $\ell$. The $i$-th realization of $\X_\ell$ is denoted by $x_\ell^i$. We also use the notation $x_\ell^{i:j}$ as a shorthand for a vector $(x_\ell^i,..., x_\ell^j)$, which is a realization of RVs $\X_\ell^{i:j}=(\X_\ell^i,..., \X_\ell^j)$. Similarly, $x_{\ell:j}^i$ will denote the realization of the $i$-th RVs from level $\ell$ to level $j$, i.e., of $\X_{\ell:j}^i=(\X_\ell^i,..., \X_j^i)$. \color{black}For a set $\mathcal{I}$, $\mathcal{I}^c$ denotes its compliment set, and $|\mathcal{I}|$ represents its cardinality. For an integer $N$, $[N]$ will be used to denote the set of all integers from $1$ to $N$. A binary memoryless asymmetric (BMA) channel and a binary memoryless symmetric (BMS) channel will be denoted by $W$ and $\widetilde{W}$, respectively. Following the notation of \cite{arikan2009channel}, we denote $N$ independent uses of channel $W$ by $W^N$. By channel combining and splitting, we get the combined channel $W_N$ and the $i$-th subchannel $W_N^{(i)}$. Specifically, for a channel $W_\ell$ at level $\ell$, $W_{\ell}^N$, $W_{\ell,N}$ and $W_{\ell}^{(i,N)}$ are used to denote its $N$ independent expansion, the combined channel and the $i$-th subchannel after polarization. An indicator function is represented by $\mathds{1}(\cdot)$. Throughout this paper, we use the binary logarithm, denoted by log, and information is measured in bits.

\section{Preliminaries of Lattice Codes}
\subsection{Definitions}
\label{sec:deftn}
A lattice is a discrete subgroup of $\mathbb{R}^{n}$ which can be described by
\begin{eqnarray}
\Lambda=\{ \lambda={B}{x}:{x}\in\mathbb{Z}^{n}\}, \notag\
\end{eqnarray}
where ${B}$ is an $n$-by-$n$ lattice generator matrix and we always assume that it has full rank in this paper.

For a vector ${x}\in\mathbb{R}^{n}$, the nearest-neighbor quantizer associated with $\Lambda$ is $Q_{\Lambda}({x})=\text{arg}\min\limits_{ \lambda\in\Lambda}\|\lambda-{x}\|$. We define the modulo lattice operation by ${x} \text{ mod }\Lambda\triangleq {x}-Q_{\Lambda}({x})$. The Voronoi region of $\Lambda$, defined by $\mathcal{V}(\Lambda)=\{{x}:Q_{\Lambda}({x})=0\}$, specifies the nearest-neighbor decoding region. The Voronoi cell is one example of fundamental region of the lattice. A measurable set $\mathcal{R}(\Lambda)\subset\mathbb{R}^{n}$ is a fundamental region of the lattice $\Lambda$ if $\cup_{\lambda\in\Lambda}(\mathcal{R}(\Lambda)+\lambda)=\mathbb{R}^{n}$ and if $(\mathcal{R}(\Lambda)+\lambda)\cap(\mathcal{R}(\Lambda)+\lambda')$ has measure 0 for any $\lambda\neq\lambda'$ in $\Lambda$. The volume of a fundamental region is equal to that of the Voronoi region $\mathcal{V}(\Lambda)$, which is given by $\text{Vol}(\Lambda)=|\text{det}({B})|$.

The theta series of $\Lambda$ (see, e.g., \cite[p.70]{yellowbook}) is defined as
\begin{eqnarray}
\Theta_{\Lambda}(\tau)=\sum_{\lambda\in\Lambda}e^{-\pi\tau\parallel\lambda\parallel^{2}}, \quad \tau>0.
\notag\
\end{eqnarray}

In this paper, the reliability condition for Bob is measured by the block error probability $P_{e}(\Lambda,\sigma^2)$ of lattice decoding. It is the probability $\text{Pr}\{{x}\notin \mathcal{V}(\Lambda)\}$ that an $n$-dimensional independent and identically distributed (i.i.d.) Gaussian noise vector ${x}$ with zero mean and variance $\sigma^{2}$ per dimension falls outside the Voronoi region $\mathcal{V}(\Lambda)$. For an $n$-dimensional lattice $\Lambda$, define the volume-to-noise ratio (VNR) of $\Lambda$ by
\begin{eqnarray}
\gamma_{\Lambda}(\sigma)\triangleq\frac{\text{Vol}(\Lambda)^\frac{2}{n}}{\sigma^2}. \notag\
\end{eqnarray}
Then we introduce the notion of lattices which are good for the AWGN channel without power constraint.
\begin{deft}[AWGN-good lattices]\label{deft:awgngood}
A sequence of lattices $\Lambda_b$ of increasing dimension $n$ is AWGN-good if, for any fixed $P_{e}(\Lambda_b,\sigma^2)\in(0,1)$, $\lim_{n\to\infty}\gamma_{\Lambda_b}(\sigma)=2\pi e$, and if, for any fixed VNR greater than $2\pi e$,
\begin{eqnarray}
\lim_{n\to\infty} P_e(\Lambda_b,\sigma^2) = 0. \notag\
\end{eqnarray}
\end{deft}

It is worth mentioning here that we do not insist on exponentially vanishing error probabilities, unlike Poltyrev's original treatment of good lattices for coding over the AWGN channel \cite{Poltyrev}. This is because a sub-exponential or polynomial decay of the error probability is often good enough.



Next, we introduce the notion of secrecy-good lattices. For this purpose, we need the capacity $C(\Lambda_e, \sigma^2)$ of the mod-$\Lambda_e$ channel, which will be defined in \eqref{eqn:modcapacity}.

\begin{deft}[Secrecy-good lattices]\label{deft:secrecygood}
A sequence of lattices $\Lambda_e$ of increasing dimension $n$ is secrecy-good if, for any fixed VNR of $\Lambda_e$ smaller than $2 \pi e$, the channel capacity $C(\Lambda_e, \sigma^2)$  vanishes:
\begin{eqnarray}
\lim_{n\to\infty} C(\Lambda_e, \sigma^2)=0. \notag\
\end{eqnarray}
\end{deft}

Note that this definition is different from that in \cite{cong2}, which is based on the flatness factor associated with the lattice Gaussian distribution. We will show that this definition is also sufficient to guarantee vanishing information leakage (see Remark \ref{rmk:sgood}).

\color{black}

\subsection{Flatness factor and lattice Gaussian distribution}
\label{seq:ffactor}
For $\sigma>0$ and ${c}\in\mathbb{R}^{n}$, the Gaussian distribution of mean ${c}$ and variance $\sigma^{2}$ is defined as
\begin{eqnarray}
f_{\sigma,{c}}({x})=\frac{1}{(\sqrt{2\pi}\sigma)^{n}}e^{-\frac{\| {x}-{c}\|^{2}}{2\sigma^{2}}}, \notag\
\end{eqnarray}
for all ${x}\in\mathbb{R}^{n}$. For convenience, let $f_{\sigma}({x})=f_{\sigma,{0}}({x})$.

Given lattice $\Lambda$, we define the $\Lambda$-periodic function
\begin{eqnarray}
f_{\sigma,\Lambda}({x})=\sum\limits_{\lambda\in\Lambda}f_{\sigma,\lambda}({x})=\frac{1}{(\sqrt{2\pi}\sigma)^{n}}\sum\limits_{\lambda\in\Lambda}e^{-\frac{\parallel {x}-\lambda\parallel^{2}}{2\sigma^{2}}}, \notag\
\end{eqnarray}
for ${x}\in\mathbb{R}^n$.

The flatness factor is defined for a lattice $\Lambda$ as \cite{cong2}
\begin{eqnarray}
\epsilon_{\Lambda}(\sigma)\triangleq\max\limits_{{x}\in\mathcal{R}(\Lambda)}\abs{\text{Vol}(\Lambda)f_{\sigma,\Lambda}({x})-1}. \notag\
\end{eqnarray}
It can be interpreted as the maximum variation of $f_{\sigma,\Lambda}({x})$ from the uniform distribution over $\mathcal{R}(\Lambda)$.
The flatness factor can be calculated using the theta series \cite{cong2}:
\[\epsilon_{\Lambda}(\sigma)=\left(\frac{\gamma_{\Lambda}(\sigma)}{2\pi}\right)^{\frac{n}{2}}\Theta_{\Lambda}\left(\frac{1}{2\pi\sigma^{2}}\right)-1.\]

We define the \emph{discrete Gaussian distribution} over $\Lambda$
centered at ${c} \in \R^n$ as the following discrete
distribution taking values in ${ \lambda} \in \Lambda$:
\[
D_{\Lambda,\sigma,{c}}({ \lambda})=\frac{f_{\sigma,{c}}({{ \lambda}})}{f_{\sigma,{c}}(\Lambda)}, \quad \forall { \lambda} \in \Lambda,
\]
where $f_{\sigma,{c}}(\Lambda) \triangleq \sum_{{ \lambda} \in
\Lambda} f_{\sigma,{c}}({{ \lambda}})=f_{\sigma,\Lambda}({c})$. Again for convenience, we write $D_{\Lambda,\sigma}=D_{\Lambda,\sigma,{0}}$.

It is also useful to define the {discrete Gaussian distribution} over a coset of $\Lambda$, i.e., the shifted lattice $\Lambda-{c}$:
\[
D_{\Lambda-{c},\sigma}({ \lambda}-{c})=\frac{f_{\sigma}({{ \lambda}}-{c})}{f_{\sigma, {\bf c}}(\Lambda)}, \quad \forall { \lambda} \in \Lambda.
\]
Note the relation $D_{\Lambda-{c},\sigma}({ \lambda}-{c}) = D_{\Lambda,\sigma,{c}}({ \lambda})$, namely, they are a shifted version of each other.

Each component of a lattice point sampled from $D_{\Lambda-{c},\sigma}$ has an average power always less than $\sigma^2$ by the following lemma.
\begin{lem}[Average power of lattice Gaussian {\cite[Lemma 1]{LingBel13}}]\label{lem:Anvpower}
Let $x=(x_1, x_2, ..., x_n)^T \sim D_{\Lambda-{c},\sigma}$. Then, for each $1 \leq i \leq n$,
\begin{eqnarray}
E[x_i^2] \leq \sigma^2.
\end{eqnarray}
\end{lem}

If the flatness factor is negligible, the discrete
Gaussian distribution over a lattice preserves the capacity of the AWGN channel.

\begin{ther}[Mutual information of discrete Gaussian distribution {\cite[Th. 2]{LingBel13}}]\label{theorem:capacity}
Consider an AWGN channel ${\Y}={\X}+{\E}$ where the input constellation $\X$ has
a discrete Gaussian distribution $D_{\Lambda-{c},\sigma_s}$
for arbitrary ${c} \in \mathbb{R}^n$, and where the variance
of the noise $\E$ is $\sigma^2$. Let the average signal power be $P_s$ so that $\SNR=P_s/\sigma^2$, and let $\widetilde{\sigma}\triangleq \frac{\sigma_s\sigma}{\sqrt{\sigma_s^2+\sigma^2}}$. Then, if
$\varepsilon = \epsilon_{\Lambda}\left(\widetilde{\sigma}\right) < \frac{1}{2}$ and $\frac{\pi\varepsilon_t}{1-\epsilon_t}\leq \varepsilon$ where
\[
\varepsilon_t \triangleq
\left\{
  \begin{array}{ll}
    \epsilon_{\Lambda}\left(\sigma_s/\sqrt{\frac{\pi}{\pi-t}}\right), & \hbox{$t \geq 1/e$} \\
    (t^{-4}+1)\epsilon_{\Lambda}\left(\sigma_s/\sqrt{\frac{\pi}{\pi-t}}\right), & \hbox{$0< t < 1/e$}
  \end{array}
\right.
\]
the discrete Gaussian constellation results in mutual information
\begin{equation}\label{eq:lattice-capacity}
I_D \geq \frac{1}{2}\log {(1+\SNR)} - \frac{5\varepsilon}{n}
\end{equation}
per channel use. Moreover, the difference between $P_s$ and $\sigma_s^2$ is bounded by
\begin{equation}\label{eq:powergap}
\big| P_s-\sigma_s^2\big| \leq \frac{2 \pi \epsilon_t}{n(1-\epsilon)}\sigma_s^2. \notag
\end{equation}
\end{ther}

A lattice $\Lambda$ or its coset $\Lambda-{c}$ with a discrete Gaussian distribution is referred to as a \emph{good constellation} for the AWGN channel if ${\epsilon_{\Lambda}(\widetilde{\sigma})}$ is negligible \cite{LingBel13}. It is further proved in \cite{LingBel13} that the channel capacity is achieved with Gaussian shaping over an AWGN-good lattice and minimum mean square error (MMSE) lattice decoding. Following Theorem \ref{theorem:capacity}, it has been shown in \cite{polarlatticeJ} that an AWGN-good polar lattice shaped according to the discrete Gaussian distribution achieves the AWGN channel capacity with sub-exponentially vanishing error probability, which means that an explicit polar lattice satisfying the power constraint and the reliability condition for Bob is already in hand. Therefore, the next section will focus on the construction of the secrecy-good polar lattice.

\section{Polar Codes and Polar Lattices}

\subsection{Polar codes: brief review}
\label{sec:polarsym}
We firstly recall some basics of polar codes. Let $\widetilde{W}$ be a BMS channel with uniformly distributed input $\X \in \mathcal{X}=\{0,1\}$ and output $\Y \in \mathcal{Y}$. The input distribution and transition probability of $\widetilde{W}$ are denoted by $P_{\X}$ and $P_{\Y|\X}$ respectively. Let $\X^{[N]}$ and $\Y^{[N]}$ be the input and output vector of $N$ independent uses of $\widetilde{W}$. Let $N=2^m$ be the block length of polar codes for some integer $m\geq1$. The channel polarization is based on the $N$-by-$N$ transform $\U^{[N]}=\X^{[N]}G_N$, where $G_{N}=\left[\begin{smallmatrix}1&0\\1&1\end{smallmatrix}\right]^{\otimes
m}$ is the generator matrix and $\otimes$ denotes the Kronecker product. Then we get an $N$-dimensional combined channel $\widetilde{W}_N$ from $\U^{[N]}$ to $\Y^{[N]}$. For each $i \in [N]$, given the previous bits $\U^{1:i-1}$, the channel $\widetilde{W}_N^{(i)}$ seen by each bit $\U^i$ is called the $i$-th subchannel channel after the channel splitting process \cite{arikan2009channel}, and the transition probability of $\widetilde{W}_N^{(i)}$ is given by
\begin{eqnarray}
\widetilde{W}_{N}^{(i)}(y^{[N]},u^{1:i-1}|u^{i})=\sum\limits_{u^{i+1:N}\in \mathcal{X}^{N-i}}\frac{1}{2^{N-1}}\widetilde{W}_{N}(y^{[N]}|u^{[N]}), \notag\
\end{eqnarray}
where $u^{[N]}$ and $y^{[N]}$ are the realizations of $\U^{[N]}$ and $\Y^{[N]}$, respectively. Ar{\i}kan proved that $\widetilde{W}_{N}^{(i)}$ is also a BMS channel and it becomes either an almost error-free channel or a completely useless channel as $N$ grows. According to \cite{arikan2009channel}, the goodness of a BMS channel can be estimated by its associate Bhattacharyya parameter, which is defined as follows.
\begin{deft}[Bhattacharyya parameter of BMS channels]
\label{deft:symmZ}
Let $\widetilde{W}$ be a BMS channel with transition probability $P_{\Y|\X}$, the symmetric Bhattacharyya parameter $\widetilde{Z}\in[0,1]$ is defined as
\begin{eqnarray}
\widetilde{Z}(\widetilde{W})&\triangleq\sum\limits_{y} \sqrt{P_{\Y|\X}(y|0)P_{\Y|\X}(y|1)}. \notag\
\end{eqnarray}
\end{deft}

\begin{rem}
Although polar codes were originally proposed for binary-input discrete memoryless channels \cite{arikan2009channel}, their extension to continuous channels, such as the binary-input AWGN channel, was given in \cite{Ido}. To construct polar codes efficiently, the authors proposed smart channel degrading and upgrading merging algorithms to quantize continuous channels into their discrete versions. Fortunately, the quantization accuracy can be made arbitrarily small by increasing the quantization level. For this reason, we still use the summation form of Bhattacharyya parameters for continuous channels in this work, which also makes the notations consistent with the literature on polar codes.
\end{rem}

It was further shown in \cite{arikan2009rate, polarchannelandsource} that for any $0<\beta<\frac{1}{2}$,
\begin{eqnarray}
\lim_{m\rightarrow \infty}\frac{1}{N}\left|\{i:\widetilde{Z}(\widetilde{W}_{N}^{(i)})<2^{-N^{\beta}}\}\right|&=&I(\widetilde{W}) \notag \\
\lim_{m\rightarrow \infty}\frac{1}{N}\left|\{i:\widetilde{Z}(\widetilde{W}_{N}^{(i)})>1-2^{-N^{\beta}}\}\right|&=&1-I(\widetilde{W}), \notag
\end{eqnarray}
which means the proportion of such roughly error-free subchannels (with negligible Bhattacharyya parameters) approaches the channel capacity $I(\widetilde{W})$. The set of the indices of all those almost error-free subchannels is usually called the information set $\mathcal{I}$ and its complementary is called the frozen set $\mathcal{F}$. Consequently, the construction of capacity-achieving polar codes is simply to identify the indices in the information set $\mathcal{I}$. However, for a general BMS channel other than binary erasure channel, the complexity of the exact computation for $\widetilde{Z}(\widetilde{W}_{N}^{(i)})$ appears to be exponential in the block length $N$. An efficient estimation method for $\widetilde{Z}(\widetilde{W}_{N}^{(i)})$ was proposed in \cite{Ido}, using the idea of channel upgrading and degrading. It was shown that with a sufficient number of quantization levels, the approximation error is negligible even if $\widetilde{W}$ has continuous output, and the involved computational complexity is acceptable.

In \cite{arikan2009channel}, a bit-wise decoding method called successive cancellation (SC) decoding was proposed to show that polar codes are able to achieve channel capacity with vanishing error probability. \color{black}This decoding method has complexity $O(N\text{log}N)$, and the error probability is given by $P_{e}^{SC} \leq \sum_{i \in \mathcal{I}} \widetilde{Z}(\widetilde{W}_{N}^{(i)})$.

\subsection{Polar codes for the binary symmetric wiretap channel}
\label{sec:polarswtp}
Now we revisit the construction of polar codes for the binary symmetric wiretap channel. We use $\widetilde{V}$ and $\widetilde{W}$ to denote the symmetric main channel between Alice and Bob and the symmetric wiretap channel between Alice and Eve, respectively. Both $\widetilde{V}$ and $\widetilde{W}$ have binary input $\X$ and $\widetilde{W}$ is degraded with respect to $\widetilde{V}$. Let $\Y$ and $\Z$ denote the output of $\widetilde{V}$ and $\widetilde{W}$. After the channel combination and splitting of $N$ independent uses of the $\widetilde{V}$ and $\widetilde{W}$ by the polarization transform $\U^{[N]}=\X^{[N]}G_N$, we define the sets of reliability-good indices for Bob and information-poor indices for Eve as
\begin{eqnarray}\label{eqn:Good&bad}
\begin{aligned}
\mathcal{G}(\widetilde{V})&=\{i:\widetilde{Z}(\widetilde{V}_N^{(i)}) \leq 2^{-N^\beta}\}, \\
\mathcal{N}(\widetilde{W})&=\{i:\widetilde{Z}(\widetilde{W}_N^{(i)}) \geq 1-2^{-N^\beta}\},
\end{aligned}
\end{eqnarray}
where $0<\beta<0.5$ and $\widetilde{V}_N^{(i)}$ ($\widetilde{W}_N^{(i)}$) is the $i$-th subchannel of the main channel (wiretapper's channel) after polarization transform.

Note that in the seminal paper \cite{polarsecrecy} of polar wiretap coding, the information-poor set $\mathcal{N}(\widetilde{W})$ was defined as $\{i:I(\widetilde{W}^{(i,N)}) \leq 2^{-N^\beta}\}$. In contrast, our criterion here is based on the Bhattacharyya parameter\footnote{This idea has already been used in \cite{polarsecrecy} to prove that polar wiretap coding scheme is secrecy capacity-achieving.}. This slight modification will bring us much convenience when lattice shaping is involved in Sect. \ref{sec:SecrecyGoodShap}. The following lemma shows that the modified criterion is similar to the original one in the sense that the mutual information of the subchannels with indices in $\mathcal{N}(\widetilde{W})$ can still be bounded in the same form.

\begin{lem}\label{lem:Mutualbound}
Let $\widetilde{W}_N^{(i)}$ be the $i$-th subchannel after the polarization transform on independent $N$ uses of a BMS channel $\widetilde{W}$. For any $0<\beta<\frac{1}{2}$ and $\delta>0$, if $\widetilde{Z}(\widetilde{W}_N^{(i)}) \geq 1-2^{-N^\beta}$, the mutual information of the $i$-th subchannel can be upper-bounded as
\begin{eqnarray}\label{eqn:Ibound}
I(\widetilde{W}_N^{(i)}) \leq 2^{-N^{\beta'}}, \notag\
\end{eqnarray}
where $\beta(1-\delta) \leq \beta' \leq \beta$ when $N$ is sufficiently large.
\end{lem}
\begin{proof}
When $\widetilde{W}$ is symmetric, $\widetilde{W}_N^{(i)}$ is symmetric as well. By \cite[Proposition 1]{arikan2009channel}, we have
\begin{eqnarray}
\begin{aligned}
I(\widetilde{W}_N^{(i)}) &\leq \sqrt{1-\widetilde{Z}(\widetilde{W}_{N}^{(i)})^2} \notag \\
 &\leq \sqrt{2\cdot 2^{-N^\beta}} \notag \\
 &= 2^{-N^{\beta'}},
\end{aligned}
\end{eqnarray}
where $\beta'<\beta$. Moreover, for sufficiently large $N$, $\beta'$ can be made arbitrarily close to and $\beta$ , i.e., $\beta(1-\delta) \leq \beta'$ for any $\delta>0$.
\end{proof}
\color{black}

Since the mutual information of subchannels in $\mathcal{N}(\widetilde{W})$ can be upper-bounded in the same form, it is not difficult to understand that strong secrecy can be achieved using the index partition proposed in \cite{polarsecrecy}. Similarly, we divide the index set $[N]$ into the following four sets:
\begin{eqnarray}\label{eqn:partition}
\begin{aligned}
&\mathcal{A}=\mathcal{G}(\widetilde{V})\cap \mathcal{N}(\widetilde{W}), \,\,\,\, \mathcal{B}=\mathcal{G}(\widetilde{V})\cap \mathcal{N}(\widetilde{W})^{c}\\
&\mathcal{C}=\mathcal{G}(\widetilde{V})^{c}\cap \mathcal{N}(\widetilde{W}), \,\,\mathcal{D}=\mathcal{G}(\widetilde{V})^{c}\cap \mathcal{N}(\widetilde{W})^{c}.
\end{aligned}
\end{eqnarray}
Clearly, $\mathcal{A} \cup \mathcal{B} \cup \mathcal{C} \cup \mathcal{D}=[N]$. Then we assign set $\mathcal{A}$ with message bits $\M$, set $\mathcal{B}$ with uniformly random bits $\mathsf{R}_b$, set $\mathcal{C}$ with frozen bits $\F$ which are known to both Bob and Eve prior to transmission, and set $\mathcal{D}$ with uniformly random bits $\mathsf{R}_d$. The next lemma shows that this assignment achieves strong secrecy. We note that this proof is similar to that in \cite{NewPolarSchemeWiretap, polarsecrecy} and it is given in \cite[Appendix A]{PolarWiretapJ}.
\begin{lem}\label{lem:strongbound}
According to the partitions of the index set shown in \eqref{eqn:partition}, if we assign the four sets as follows
\begin{eqnarray}\label{eqn:assign}
\begin{aligned}
&\mathcal{A}\leftarrow \M, \,\,\,\,\, \mathcal{B}\leftarrow \mathsf{R}_b,\\
&\mathcal{C}\leftarrow \F, \,\,\,\,\,\, \mathcal{D}\leftarrow \mathsf{R}_d,
\end{aligned}
\end{eqnarray}
the information leakage $I(\M;\Z^{[N]})$ can be upper-bounded as
\begin{eqnarray}\label{eqn:leakageBound}
I(\M;\Z^{[N]}) \leq N\cdot 2^{-N^{\beta'}}, 0<\beta'<0.5.
\end{eqnarray}
\end{lem}

We can also observe that the proportion of the problematic set $\mathcal{D}$ is arbitrarily small when $N$ is sufficiently large. This is because set $\mathcal{D}$ is a subset of the unpolarized set $\{i: 2^{-N^\beta}<\widetilde{Z}(\widetilde{V}_{N}^{(i)}) < 1-2^{-N^\beta}\}$. As has been shown in \cite{polarsecrecy}, the reliability condition cannot be fulfilled with SC decoding due to the existence of $\mathcal{D}$. Fortunately, we can use the Markov block coding technique proposed in \cite{NewPolarSchemeWiretap} to achieve reliability and strong secrecy simultaneously. More details of this Markov block coding technique will be discussed in Section \ref{sec:reliability} and Section \ref{sec:reliabilityshape}.

With regard to the secrecy rate, we show that the modified polar coding scheme can also achieve the secrecy capacity.
\begin{lem}\label{lem:securerate}
Let $C(\widetilde{V})$ and $C(\widetilde{W})$ denote the channel capacity of the main channel $\widetilde{V}$ and wiretap channel $\widetilde{W}$ respectively. Since $\widetilde{W}$ is degraded with respect to $\widetilde{V}$, the secrecy capacity, which is given by $C(\widetilde{V})-C(\widetilde{W})$, is achievable using the modified wiretap coding scheme, i.e.,
\begin{eqnarray}\label{eqn:sececrycap}
\lim_{N\rightarrow \infty} |\mathcal{G}(\widetilde{V})\cap \mathcal{N}(\widetilde{W})|/N=C(\widetilde{V})-C(\widetilde{W}). \notag
\end{eqnarray}
\end{lem}
\begin{proof}
See \cite[Appendix B]{PolarWiretapJ}.
\end{proof}
\color{black}

\subsection{From polar codes to polar lattices}

A sublattice $\Lambda' \subset \Lambda$ induces a partition (denoted by $\Lambda/\Lambda'$) of $\Lambda$ into equivalence classes modulo $\Lambda'$. The order of the partition is denoted by $|\Lambda/\Lambda'|$, which is equal to the number of cosets. If $|\Lambda/\Lambda'|=2$, we call this a binary partition. Let $\Lambda/\Lambda_{1}/\cdots/\Lambda_{r-1}/\Lambda'$ for $r \geq 1$ be an $n$-dimensional self-similar lattice partition chain\footnote{By saying self-similar, we mean that $\Lambda_\ell=T^\ell \Lambda$ for all $\ell$, with $T=\alpha V$ for some scale factor $\alpha>1$ and orthogonal matrix $V$. For example, $\mathbb{Z}/2\mathbb{Z}/.../2^{r}\mathbb{Z}$ is a one-dimensional self-similar partition chain.}.  For each
partition $\Lambda_{\ell-1}/\Lambda_{\ell}$ ($1\leq \ell \leq r$ with convention $\Lambda_0=\Lambda$ and $\Lambda_r=\Lambda'$) a code $C_{\ell}$ over $\Lambda_{\ell-1}/\Lambda_{\ell}$
selects a sequence of representatives $a_{\ell}$ for the cosets of $\Lambda_{\ell}$. Consequently, if each partition is binary, the code $C_{\ell}$ is a binary code.

Polar lattices are constructed by ``Construction D" \cite[p.232]{yellowbook} \cite{forney6} using a set of nested polar codes $C_{1}\subseteq C_{2}\cdot\cdot\cdot\subseteq C_{r}$. \color{black}
Suppose $C_{\ell}$ has block length $N$ and $k_{\ell}$ information bits for $1\leq\ell\leq r$. Choose a basis $\mathbf{g}_{1}, \mathbf{g}_{2},\cdots, \mathbf{g}_{N}$ from the polar generator matrix $G_N$ such that $\mathbf{g}_{1},\cdots \mathbf{g}_{k_{\ell}}$ span $C_{\ell}$. When the dimension $n=1$, we choose the partition chain $\mathbb{Z}/2\mathbb{Z}.../2^r\mathbb{Z}$, then the lattice $L$ admits the form \cite{forney6}
\begin{eqnarray}
L = \left\{\sum_{\ell=1}^{r}2^{\ell-1}\sum_{i=1}^{k_{\ell}}u_{\ell}^{i}\mathbf{g}_{i}+2^{r}\mathbb{Z}^N\mid u_{\ell}^{i}\in\{0,1\}\right\},
\label{constructionD}
\end{eqnarray}
where the addition is carried out in $\mathbb{R}^N$. The fundamental volume of a lattice obtained from this construction is given by
\begin{eqnarray}
\text{Vol}(L)=2^{-NR_{C}}\cdot \text{Vol}(\Lambda_r)^{N}, \notag\
\end{eqnarray}
where $R_{C}=\sum_{\ell=1}^{r} R_{\ell}=\frac{1}{N}\sum_{\ell=1}^{r}k_{\ell}$ denotes the sum rate of component codes. In this paper, we limit ourselves to the one-dimensional binary lattice partition chain and binary polar codes for simplicity.

\section{Secrecy-Good Polar Lattices for the Mod-$\Lambda_s$ GWC}\label{sec:NoPower}

Before considering the Gaussian wiretap channel, we will tackle a simpler problem of constructing secrecy-good polar lattices over the mod-$\Lambda_s$ GWC shown in Fig. \ref{fig:wiretap_mod}. The difference between the mod-$\Lambda_s$ GWC and the genuine GWC is the mod-$\Lambda_s$ operation on the received signal of Bob and Eve. We will assume uniform input messages until we discuss semantic security in the end of this section.

\subsection{Strong secrecy}\label{sec:SecrecyGoodNoShap}

\begin{figure}[h]
    \centering
    \includegraphics[width=12cm]{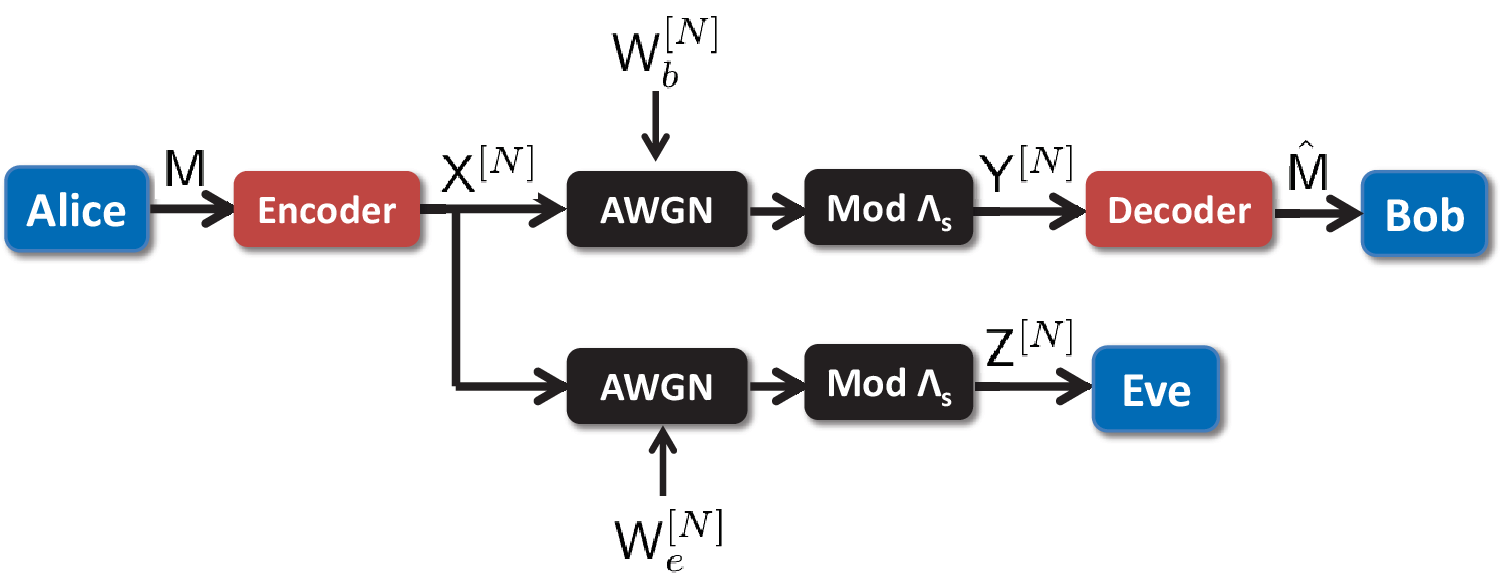}
    \caption{The mod-$\Lambda_s$ Gaussian wiretap channel.}
    \label{fig:wiretap_mod}
\end{figure}

With some abuse of notation, the outputs $\Y^{[N]}$ and $\Z^{[N]}$ at Bob and Eve's ends respectively become
\begin{eqnarray}
\left\{\begin{aligned}
&\Y^{[N]}=\Big[\X^{[N]}+\W_{b}^{[N]}\Big]\text{ mod } \Lambda_s, \notag\\
&\Z^{[N]}=\Big[\X^{[N]}+\W_{e}^{[N]}\Big] \text{ mod } \Lambda_s. \notag\
\end{aligned}\right.
\notag\
\end{eqnarray}

The idea of wiretap lattice coding over the mod-$\Lambda_s$ GWC \cite{cong2} can be explained as follows. Let $\Lambda_b$ and $\Lambda_e$ be the AWGN-good lattice and secrecy-good lattice designed for Bob and Eve accordingly. Let $\Lambda_s \subset \Lambda_e \subset \Lambda_b$ be a nested chain of $N$-dimensional lattices in $\mathbb{R}^N$, where $\Lambda_s$ is the shaping lattice. Note that the shaping lattice $\Lambda_s$ here is employed primarily for the convenience of designing the secrecy-good lattice and secondarily for satisfying the power constraint. Consider a one-to-one mapping: $\mathcal{M} \rightarrow \Lambda_b/\Lambda_e$ which associates each message $m \in \mathcal{M}$ to a coset $\widetilde{\lambda}_m \in \Lambda_b/\Lambda_e$. Alice selects a lattice point $\lambda \in \Lambda_e \cap \mathcal{V}(\Lambda_s)$ uniformly at random and transmits $\X^{[N]}=\lambda+\lambda_m$, where $\lambda_m$ is the coset representative of $\widetilde{\lambda}_m$ in $\mathcal{V}(\Lambda_e)$. This scheme has been proved to achieve both reliability and semantic security in \cite{cong2} by random lattice codes. We will make it explicit by constructing polar lattice codes in this section.

Let $\Lambda_b$ and $\Lambda_e$ be constructed from a binary partition chain $\Lambda/\Lambda_{1}/\cdots/\Lambda_{r-1}/\Lambda_r$, and assume $\Lambda_s \subset \Lambda_r^N$ such that $\Lambda_s \subset \Lambda_r^N \subset \Lambda_e \subset \Lambda_b$\footnote{This is always possible with sufficient power, since the power constraint is not our primary concern in this section. We can scale $\Lambda_s$ as large as possible to make $\Lambda_s \subset \Lambda_r^N$.}. Also, denote by $\X_{1:r}^{[N]}$ the bits encoding $\Lambda^N/\Lambda_r^N$, which include all information bits for message $\M$ as a subset. We have that $\Big[\X^{[N]}+\W_e^{[N]}\Big]$ mod $\Lambda_r^N $ is a sufficient statistic for $\X_{1:r}^{[N]}$. This can be seen from \cite[Lemma 8]{forney6}, rewritten as follows:

\begin{lem}[Sufficiency of mod-$\Lambda$ output \cite{forney6}]\label{lem:modsuffic}
For a partition chain $\Lambda/\Lambda'$ ($\Lambda' \subset \Lambda$), let the input of an AWGN channel be $\X=\mathsf{A}+\mathsf{B}$, where $\mathsf{A}\in \mathcal{R}(\Lambda)$ is a random variable, and $\mathsf{B}$ is uniformly distributed in $\Lambda \cap \mathcal{R}(\Lambda')$. Reduce the output $\mathsf{Y}$ first to $\mathsf{Y}'=\mathsf{Y} \mod \Lambda'$ and then to $\mathsf{Y}''=\mathsf{Y}' \mod \Lambda$. Then the mod-$\Lambda$ map is information-lossless, namely $I(\mathsf{A};\mathsf{Y}')=I(\mathsf{A};\mathsf{Y}'')$, which means that the output $\mathsf{Y}''=\mathsf{Y}' \mod \Lambda$ of mod-$\Lambda$ map is a sufficient statistic for $\mathsf{A}$.
\end{lem}

In our context, we identify $\Lambda$ with $\Lambda_r^N$ and $\Lambda'$ with $\Lambda_s$, respectively. Since the bits encoding $\Lambda_r^N/\Lambda_s$ are uniformly distributed\footnote{In fact, all bits encoding $\Lambda_e/\Lambda_s$ are uniformly distributed in wiretap coding.}, the mod-$\Lambda_r^N$ operation is information-lossless in the sense that \[I\Big(\X_{1:r}^{[N]};\Z^{[N]}\Big)=I\Big(\X_{1:r}^{[N]};[\X^{[N]}+\W_{e}^{[N]}]\text{ mod }\Lambda_r^N\Big).\] As far as mutual information $I\Big(\X_{1:r}^{[N]};\Z^{[N]}\Big)$ is concerned, we can use the mod-$\Lambda_r^N$ operator instead of the mod-$\Lambda_s$ operator here. Under this condition, we use the multilevel lattice structure introduced in \cite{forney6} to decompose the mod-$\Lambda_s$ channel into a series of BMS channels according to the partition chain $\Lambda/\Lambda_1/ \cdots/\Lambda_{r-1}/\Lambda_r$. \color{black}Therefore, the afore-mentioned polar coding technique for BMS channels can be employed. Moreover, the channel resulted from the lattice partition chain can be proved to be equivalent to that based on the chain rule of mutual information (See \eqref{eqn:chain_rule}). Following this channel equivalence, we can construct an AWGN-good lattice $\Lambda_b$ and a secrecy-good lattice $\Lambda_e$, using the wiretap coding technique \eqref{eqn:Good&bad} at each partition level.

A mod-$\Lambda$ channel is a Gaussian channel with a modulo-$\Lambda$ operator in the front end \cite{multilevel1,forney6}. The capacity of the mod-$\Lambda$ channel is \cite{forney6}
\begin{eqnarray}\label{eqn:chpt2_mod-capacity}
C(\Lambda, \sigma^{2})=\log (\text{Vol}(\Lambda))-h(\Lambda, \sigma^{2}),
\label{eqn:modcapacity}
\end{eqnarray}
where $h(\Lambda, \sigma^{2})$ is the differential entropy of the $\Lambda$-aliased noise over $\mathcal{R}(\Lambda)$:
\begin{eqnarray}
h(\Lambda,\sigma^{2})=-\int_{\mathcal{R}(\Lambda)}f_{\sigma,\Lambda}(t)\text{ log } f_{\sigma,\Lambda}(t)dt. \notag\
\end{eqnarray}
The differential entropy reaches its maximum $\log (\mathrm{Vol}(\Lambda))$ by the uniform distribution over $\mathcal{R}(\Lambda)$. The $\Lambda_{\ell-1}/\Lambda_{\ell}$ channel is defined as a mod-$\Lambda_{\ell}$ channel whose input is drawn from $\Lambda_{\ell-1}\cap\mathcal{R}(\Lambda_{\ell})$. It is known that the $\Lambda_{\ell-1}/\Lambda_{\ell}$ channel is symmetric\footnote{This is ``regular" in the sense of Delsarte and Piret and symmetric in the sense of Gallager \cite{forney6}.}, and the optimum input distribution is uniform \cite{forney6}. Furthermore, the $\Lambda_{\ell-1}/\Lambda_{\ell}$ channel is binary if $|\Lambda_{\ell-1}/\Lambda_{\ell}|=2$. The capacity of the $\Lambda_{\ell-1}/\Lambda_{\ell}$ channel for Gaussian noise of variance $\sigma^2$ is given by \cite{forney6}
\begin{eqnarray}\notag
\begin{split}
  C(\Lambda_{\ell-1}/\Lambda_{\ell}, \sigma^2) &= C(\Lambda_{\ell}, \sigma^2) - C(\Lambda_{\ell-1}, \sigma^2) \\
  &= h(\Lambda_{\ell-1}, \sigma^2) - h(\Lambda_{\ell}, \sigma^2) + \log (\text{Vol}(\Lambda_{\ell})/\text{Vol}(\Lambda_{\ell-1})).
\end{split}
\end{eqnarray}

The decomposition into a set of $\Lambda_{\ell-1}/\Lambda_{\ell}$ channels is used in \cite{forney6} to construct AWGN-good lattices. Take the partition chain $\mathbb{Z}/2\mathbb{Z}/\cdot\cdot\cdot/2^r\mathbb{Z}$ as an example. Given uniform input $\X_{1:r}$, let $\mathcal{K}_\ell$ denote the coset indexed by $x_{1:\ell}$, i.e., $\mathcal{K}_{\ell}=x_1+\cdot\cdot\cdot+2^{\ell-1}x_{\ell}+2^{\ell}\mathbb{Z}$. Given that $\X_{1:\ell-1}=x_{1:\ell-1}$, the conditional probability distribution function (PDF) of this channel with binary input $\X_\ell$ and output $\bar{\Z}=\Z \mod \Lambda_{\ell}$ is
\begin{eqnarray}\label{eqn:modchannel}
f_{\bar{\Z}|\X_\ell}(\bar{z}|x_\ell)=\frac{1}{\sqrt{2\pi}\sigma_e}\sum\limits_{a\in \mathcal{K}_\ell(x_{1:\ell})}\text{exp}\left(-\frac{1}{2\sigma_e^{2}}\|\bar{z}-a\|^2\right).
\end{eqnarray}
Since the previous input bits $x_{1:\ell-1}$ cause a shift on $\mathcal{K}_\ell$ and will be removed by the multistage decoder at level $\ell$, the code can be designed according to the channel transition probability \eqref{eqn:modchannel} with $x_{1:\ell-1}=0$. Following the notation of \cite{forney6}, we use $V(\Lambda_{\ell-1}/\Lambda_{\ell},\sigma_b^2)$ and $W(\Lambda_{\ell-1}/\Lambda_{\ell},\sigma_e^2)$ to denote the $\Lambda_{\ell-1}/\Lambda_{\ell}$ channel for Bob and Eve respectively. The $\Lambda_{\ell-1}/\Lambda_{\ell}$ channel can also be used to construct secrecy-good lattices. In order to bound the information leakage of the wiretapper's channel, we firstly express $I(\X_{1:r};\Z)$ according to the chain rule of mutual information as
\begin{eqnarray}\label{eqn:chain_rule}
\begin{aligned}
I(\X_{1:r};\Z) =I(\X_1;\Z)+I(\X_2;\Z|\X_1)+\cdot\cdot\cdot +I(\X_{r};\Z|\X_{1:r-1}).
\end{aligned}
\end{eqnarray}
This equation still holds if $\Z$ denotes the noisy signal after the mod-$\Lambda_r$ operation, namely, $\Z=[\X+\W_e] \text{ mod } \Lambda_r$. We will adopt this notation in the rest of this subsection.
We refer to the $\ell$-th channel associated with mutual information $I(\X_{\ell};\Z|\X_{1:\ell-1})$ as the equivalent channel denoted by $W'(\X_{\ell};\Z|\X_{1:\ell-1})$, which is defined as the channel from $\X_\ell$ to $\Z$ given the previous $\X_{1:\ell-1}$. Then the transition probability distribution of $W'(\X_{\ell};\Z|\X_{1:\ell-1})$ is \cite[Lemma 6]{forney6}
{\allowdisplaybreaks\begin{eqnarray}\label{eqn:modchannelcapacity}
\begin{aligned}
f_{\Z|\X_\ell}(z|x_\ell)&=\frac{1}{\text{Pr}(\mathcal{K}_\ell(x_{1:\ell}))}\sum_{a\in\mathcal{K}_\ell(x_{1:\ell})}\text{Pr}(a)f_{\Z}(z|a) \\
&=\frac{1}{|\Lambda_{\ell}/\Lambda_r|}\frac{1}{\sqrt{2\pi}\sigma_e}\sum\limits_{a\in \mathcal{K}_\ell(x_{1:\ell})}\text{exp}\left(-\frac{1}{2\sigma_e^{2}}\|z-a\|^2\right),\:\: z\in\mathcal{V}(\Lambda_r).
\end{aligned}
\end{eqnarray}}

From \eqref{eqn:modchannel} and \eqref{eqn:modchannelcapacity}, we can observe that the channel output likelihood ratio (LR) of the $W(\Lambda_{\ell-1}/\Lambda_{\ell},\sigma_e^2)$ channel is equal to that of the $\ell$-th equivalent channel $W'(\X_{\ell};\Z|\X_{1:\ell-1})$. Then we have the following channel equivalence lemma.

\begin{lem}\label{lem:channelequ}
Consider a lattice $L$ constructed by a binary lattice partition chain $\Lambda/\Lambda_{1}/\cdots/\Lambda_{r-1}/\Lambda_r$. Constructing a polar code for the $\ell$-th equivalent binary-input channel $W'(\X_{\ell};\Z|\X_{1:\ell-1})$ defined by the chain rule \eqref{eqn:chain_rule} is equivalent to constructing a polar code for the $\Lambda_{\ell-1}/\Lambda_{\ell}$ channel $W(\Lambda_{\ell-1}/\Lambda_{\ell},\sigma_e^2)$, i.e., the mutual information and Bhattacharyya parameters of the subchannels resulted from $W'(\X_{\ell};\Z|\X_{1:\ell-1})$ are equivalent to that of the subchannels resulted from $W(\Lambda_{\ell-1}/\Lambda_{\ell},\sigma_e^2)$, respectively.
\end{lem}
\begin{proof}
See Appendix \ref{appendix2}.
\end{proof}
Note that another proof based on direct calculation of the mutual information and Bhattacharyya parameters of the subchannels can be found in \cite{yan3}.
\begin{rem}
Observe that if we define $V'(\X_\ell;\Y|\X_{1:\ell-1})$ as the equivalent channel according to the chain rule expansion of $I(\X;\Y)$ for the main channel, the same result can be obtained between $V(\Lambda_{\ell-1}/\Lambda_{\ell},\sigma_b^2)$ and $V'(\X_\ell;\Y|\X_{1:\ell-1})$. Moreover, this lemma also holds without the mod-$\Lambda_s$ front-end, i.e., without power constraint. The construction of AWGN-good polar lattices was given in \cite{polarlatticeJ}, where nested polar codes were constructed based on a set of $\Lambda_{\ell-1}/\Lambda_{\ell}$ channels. We note that the $\Lambda_{\ell-1}/\Lambda_{\ell}$ channel is degraded with respect to the $\Lambda_{\ell}/\Lambda_{\ell+1}$ channel \cite[Lemma 3]{polarlatticeJ}.
\end{rem}

\begin{figure*}[htp]
    \centering
    \includegraphics[width=15cm,height=6cm,keepaspectratio]{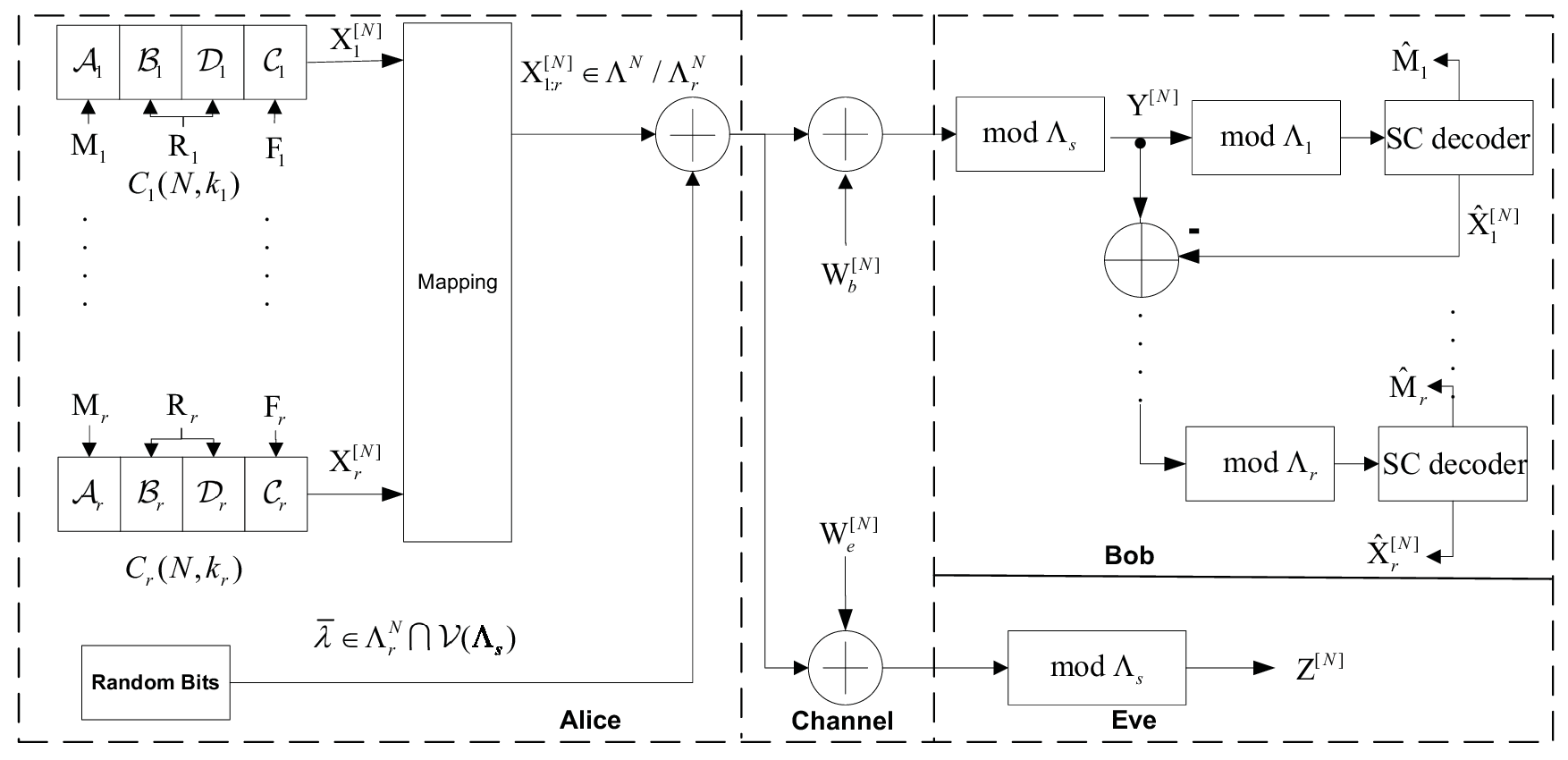}
    \caption{The multilevel lattice coding system over the mod-$\Lambda_s$ Gaussian wiretap channel.}
    \label{fig:eve}
\end{figure*}

Now we are ready to introduce the polar lattice construction for the mod-$\Lambda_s$ GWC shown in Fig. \ref{fig:eve}. A polar lattice $L$ is constructed by a series of nested polar codes $C_{1}(N,k_{1})\subseteq C_{2}(N,k_{2})\subseteq\cdot\cdot\cdot\subseteq C_{r}(N,k_{r})$ and a binary lattice partition chain $\Lambda/\Lambda_1/\cdot\cdot\cdot/\Lambda_r$. The block length of polar codes is $N$. Alice splits the message $\M$ into $\M_1,\cdot\cdot\cdot,\M_{r}$. We follow the same rule \eqref{eqn:assign} to assign bits in the component polar codes to achieve strong secrecy. Note that $W(\Lambda_{\ell-1}/\Lambda_{\ell},\sigma_e^2)$ is degraded with respect to $V(\Lambda_{\ell-1}/\Lambda_{\ell},\sigma_b^2)$ for $1\leq \ell\leq r$ because $\sigma_b^2 \leq \sigma_e^2$. Treating $V(\Lambda_{\ell-1}/\Lambda_{\ell},\sigma_b^2)$ and $W(\Lambda_{\ell-1}/\Lambda_{\ell},\sigma_e^2)$ as the main channel and wiretapper's channel at each level and using the partition rule \eqref{eqn:partition}, we can get four sets $\mathcal{A}_\ell$, $\mathcal{B}_\ell$, $\mathcal{C}_\ell$ and $\mathcal{D}_\ell$. Similarly, we assign the bits as follows
\begin{eqnarray}\label{eqn:newassign}
\begin{aligned}
&\mathcal{A}_\ell\leftarrow \M_\ell, \;\;\mathcal{B}_\ell\leftarrow \mathsf{R}^b_\ell,\\
&\mathcal{C}_\ell \leftarrow \F_\ell, \;\;\;\mathcal{D}_\ell \leftarrow \mathsf{R}^d_\ell
\end{aligned}
\end{eqnarray}
for each level $\ell$, where $\M_\ell$, $\F_\ell$ and $\mathsf{R}^b_\ell$ ($\mathsf{R}^d_\ell$) represent message bits, frozen bits (could be set as all zeros) and uniformly random bits for set $\mathcal{B}_\ell$ ($\mathcal{D}_\ell$) at level $\ell$. Since the $\Lambda_{\ell-1}/\Lambda_{\ell}$ channel is degraded with respect to the $\Lambda_{\ell}/\Lambda_{\ell+1}$ channel. According to \cite[Lemma 4.7]{polarchannelandsource}, when a BMS channel $\widetilde{W}$ is degraded with respect to a BMS channel $\widetilde{V}$, the Bhattacharyya parameters of the subchannels satisfy $\widetilde{Z}(\widetilde{W}_N^{(i)}) \geq \widetilde{Z}(\widetilde{V}_N^{(i)})$. Thus, it is easy to obtain that $\mathcal{C}_\ell\supseteq\mathcal{C}_{\ell+1}$, which means $\mathcal{A}_{\ell}\cup\mathcal{B}_{\ell}\cup\mathcal{D}_\ell\subseteq\mathcal{A}_{\ell+1}\cup\mathcal{B}_{\ell+1}\cup\mathcal{D}_{\ell+1}$. This construction is clearly a lattice construction as polar codes constructed for each level are nested.
We skip the proof of nested polar codes here. A similar proof can be found in \cite{yan2} and \cite{polarlatticeJ}.

As a result, the above multilevel construction yields an AWGN-good lattice $\Lambda_b$ and a secrecy-good lattice $\Lambda_e$ simultaneously. More precisely, $\Lambda_b$ is constructed from a set of nested polar codes $C_{1}(N,|\mathcal{A}_1|+|\mathcal{B}_1|+|\mathcal{D}_1|)\subseteq\cdot\cdot\cdot\subseteq C_{r}(N,|\mathcal{A}_{r}|+|\mathcal{B}_{r}|+|\mathcal{D}_{r}|)$, while $\Lambda_e$ is constructed from a set of nested polar codes $C_{1}(N,|\mathcal{B}_1|+|\mathcal{D}_1|)\subseteq\cdot\cdot\cdot\subseteq C_{r}(N,|\mathcal{B}_{r}|+|\mathcal{D}_{r}|)$ and with the same lattice partition chain. Note that the random bits in set $\mathcal{D}_\ell$ should be shared to Bob to guarantee the AWGN-goodness of $\Lambda_b$. More details will be given in the next subsection. It is clear that $\Lambda_e\subset\Lambda_b$. Thus, our proposed coding scheme instantiates the coset coding scheme introduced in \cite{cong2}, where the confidential message is mapped to the coset $\widetilde{\lambda}_m\in\Lambda_b /\Lambda_e$. However, unlike the work of \cite{cong2}, our scheme does not require an asymptotically vanishing flatness factor, since the upper-bound of the information leakage can be calculated directly. The flatness factor will show up with the lattice Gaussian shaping in the next section.

By using the above assignments and Lemma \ref{lem:strongbound}, we have
\begin{eqnarray}\label{eqn:modchannelbound}
\begin{aligned}
I\Big(\M_\ell\F_\ell;\Z_\ell^{[N]}\Big)\leq N2^{-N^{\beta'}},
\end{aligned}
\end{eqnarray}
where $\Z_\ell^{[N]}=\Z^{[N]} \text{ mod } \Lambda_{\ell}$ is the output of the $\Lambda_{\ell-1}/\Lambda_{\ell}$ channel for Eve. In other words, the employed polar code for the channel $W(\Lambda_{\ell-1}/\Lambda_{\ell},\sigma_e^2)$ can guarantee that the mutual information between the input message and the output is upper bounded by $N2^{-N^{\beta'}}$.

We assume uniform $\M_\ell$ and $\F_\ell$ such that $\X_\ell$ is uniformly distributed at each level. We will remove this restriction to the uniform distribution in Proposition \ref{rmk:semtic}. According to Lemma \ref{lem:channelequ}, the constructed polar code can also guarantee the same upper-bound on the mutual information between the input message and the output of the channel $W'(\X_{\ell};\Z|\X_{1:\ell-1})$, as shown in the following inequality ($\X_\ell$ is independent of the previous $\X_{1:\ell-1}$):
\begin{eqnarray}
\begin{aligned}
I\Big(\M_\ell \F_\ell;\Z^{[N]},\X^{[N]}_{1:\ell-1}\Big)\leq N2^{-N^{\beta'}}.\notag\
\end{aligned} \notag\
\end{eqnarray}

Recall that $\Z^{[N]}$ is the signal received by Eve after the mod-$\Lambda_r$ operation. Let $\F$ denote the combination of $\F_1, \F_2, ..., \F_r$. From the chain rule of mutual information, we obtain
{\allowdisplaybreaks\begin{flalign}\label{eqn:upperbound} \notag
&I\Big(\M \F;\Z^{[N]}\Big)  \\ \notag
&=\sum_{\ell=1}^{r}I\Big(\Z^{[N]};\M_\ell \F_\ell|\M_{1:\ell-1}\F_{1:\ell-1}\Big) \\ \notag
&=\sum_{\ell=1}^{r}H(\M_\ell \F_\ell|\M_{1:\ell-1}\F_{1:\ell-1})-H\Big(\M_\ell\F_{\ell}|\Z^{[N]},\M_{1:\ell-1}\F_{1:\ell-1}\Big)\\
&\leq\sum_{\ell=1}^{r}H(\M_\ell\F_{\ell})-H\Big(\M_\ell \F_\ell|\Z^{[N]},\M_{1:\ell-1}\F_{1:\ell-1}\Big)\\ \notag
&=\sum_{\ell=1}^{r}I\Big(\M_\ell\F_\ell;\Z^{[N]},\M_{1:\ell-1}\F_{1:\ell-1}\Big)\\ \notag
&\leq \sum_{\ell=1}^{r}I\Big(\M_\ell \F_\ell;\Z^{[N]},\X^{[N]}_{1:\ell-1}\Big) \leq rN2^{-N^{\beta'}}, \notag
\end{flalign}}where the second inequality holds because $I\Big(\M_\ell \F_\ell;\Z^{[N]},\X^{[N]}_{1:\ell-1}\Big)=I\Big(\M_\ell \F_\ell;\Z^{[N]},\U^{[N]}_{1:\ell-1}\Big)$ and adding more variables will not decrease the mutual information. Since $\lim_{N\rightarrow\infty}I\Big(\M\F;\Z^{[N]}\Big)=0$, strong secrecy is achieved.

\subsection{Achieving secrecy capacity}\label{sec:reliability}
In the original polar coding scheme for the binary wiretap channel \cite{polarsecrecy}, how to assign set $\mathcal{D}$ is a problem. Assigning frozen bits to $\mathcal{D}$ guarantees reliability but only achieves weak secrecy, whereas assigning random bits to $\mathcal{D}$ guarantees strong secrecy but may violate the reliability requirement because $\mathcal{D}$ may be nonempty. In order to ensure strong secrecy, $\mathcal{D}$ is assigned with random bits ($\mathcal{D} \leftarrow \mathsf{R}$), which makes this scheme failed to accomplish the theoretical reliability. In simple words, to satisfy the strong secrecy and reliability conditions simultaneously, the bits corresponding to $\mathcal{D}$ must be kept frozen to Bob but uniformly random to Eve. For any $\ell$-th level channel $V(\Lambda_{\ell-1}/\Lambda_{\ell}, \sigma_b^2)$ at Bob's end, if set $\mathcal{D}_\ell$ is fed with random bits, the probability of error is upper-bounded by the sum of the Bhattacharyya parameters $\widetilde{Z}(V_N^{(j)}(\Lambda_{\ell-1}/\Lambda_{\ell}, \sigma_b^2))$ of subchannels that are not frozen to zero \cite{arikan2009channel}. For each bit-channel index $j$ and $\beta <0.5$, we have
\begin{eqnarray*}
j \in \mathcal{G}(V(\Lambda_{\ell-1}/\Lambda_{\ell}, \sigma_b^2)) \cup \mathcal{D}_\ell.
\end{eqnarray*}

By the definition \eqref{eqn:Good&bad}, the sum of $\widetilde{Z}(V_N^{(j)}(\Lambda_{\ell-1}/\Lambda_{\ell}, \sigma_b^2))$ over the set $\mathcal{G}(V(\Lambda_{\ell-1}/\Lambda_{\ell}, \sigma_b^2))$ is bounded by $2^{-N^{\beta}}$, therefore the error probability of the $\ell$-th level channel under the SC decoding, denoted by $P_e^{SC}(\Lambda_{\ell-1}/\Lambda_{\ell}, \sigma_b^2)$, can be upper-bounded by \cite{arikan2009channel}
\begin{eqnarray*}
P_e^{SC}(\Lambda_{\ell-1}/\Lambda_{\ell}, \sigma_b^2) \leq N2^{-N^{\beta}} + \sum_{j \in \mathcal{D}_\ell} \widetilde{Z}(V_N^{(j)}(\Lambda_{\ell-1}/\Lambda_{\ell}, \sigma_b^2)).
\end{eqnarray*}
Since multistage decoding is utilized, by the union bound, the final decoding error probability for Bob is bounded as
\begin{eqnarray*}
\text{Pr}\{\widehat{\M} \neq \M\} \leq \sum_{i=1}^{r} P_e^{SC}(\Lambda_{\ell-1}/\Lambda_{\ell}, \sigma_b^2) .
\end{eqnarray*}
Unfortunately, a bound on the sum $\sum_{j \in \mathcal{D}_\ell} \widetilde{Z}(V_N^{(j)}(\Lambda_{\ell-1}/\Lambda_{\ell}, \sigma_b^2))$ is unavailable, making the proof of reliability out of reach. There is numerical evidence of low probabilities of error nonetheless. The proportion of $\mathcal{D}_\ell$ vanishes as $N \to \infty$ \cite[Prop. 22]{polarsecrecy}. In fact, numerical examples in \cite[Sect. VI-F]{polarsecrecy} showed that $\mathcal{D}_\ell=\emptyset$ in most cases of interest. In any case, Bob can run some exhaustive search or form a small list of paths for those unreliable indexes.

\color{black}

The reliability problem was recently solved in \cite{NewPolarSchemeWiretap}, where a new scheme dividing the information message into several blocks was proposed. For a specific block, $\mathcal{D}_\ell$ is still assigned with random bits and transmitted in advance in the set $\mathcal{A}_\ell$ of the previous block. This scheme involves negligible rate loss and finally realizes reliability and strong security simultaneously. In this case, if the reliability of each partition channel can be achieved, i.e., for any $\ell$-th level partition $\Lambda_{\ell-1}/\Lambda_{\ell}$, $P_e^{SC}(\Lambda_{\ell-1}/\Lambda_{\ell}, \sigma_b^2)$ vanishes as $N \rightarrow \infty$, then the total decoding error probability for Bob can be made arbitrarily small. Consequently, based on this new scheme of assigning the problematic set, the error probability on level $\ell$ can be upper-bounded by
\begin{eqnarray}
P_e^{SC}(\Lambda_{\ell-1}/\Lambda_{\ell}, \sigma_b^2) \leq \epsilon_{N'}^\ell +k_\ell \cdot O(2^{-N'^{\beta}}),
\end{eqnarray}
where $k_\ell$ is the number of information blocks on the $\ell$-th level, $N'$ is the length of each block which satisfies $N'\times k_\ell=N$ and $\epsilon_{N'}^\ell$ is caused by the first separate block consisting of the initial bits in $\mathcal{D}_\ell$ at the $\ell$-th level. Since $|\mathcal{D}_\ell|$ is extremely small comparing to the block length $N$, the decoding failure probability for the first block can be made arbitrarily small when $N$ is sufficiently large. Meanwhile, by the analysis in \cite{polarlatticeJ}, when $h(\Lambda, \sigma_b^2) \rightarrow \log(V(\Lambda))$, $h(\Lambda_r, \sigma_b^2) \rightarrow \frac{1}{2}\log(2\pi e \sigma_b^2)$, and $R_C \rightarrow C(\Lambda/\Lambda_r, \sigma_b^2)$, we have $\gamma_{\Lambda_b}(\sigma_b)\rightarrow 2\pi e.$ Therefore, $\Lambda_b$ is an AWGN-good lattice\footnote{More precisely, to make $\Lambda_b$ AWGN-good, we need $P_e(\Lambda_b, \sigma_b^2) \rightarrow 0$ by definition. By \cite[Theorem 2]{polarlatticeJ}, $P_e(\Lambda_b, \sigma_b^2) \leq rN2^{-N^\beta}+N\cdot P_e(\Lambda_r, \sigma_b^2)$. According to the analysis in Remark \ref{rmk:levelnum}, $r=O(\log N)$ is sufficient to guarantee $P_e(\Lambda_r, \sigma_b^2)=e^{-\Omega(N)}$, meaning that a sub-exponentially vanishing $P_e(\Lambda_b, \sigma_b^2)$ can be achieved.}.

Note that the rate loss incurred by repeatedly transmitted bits in $\mathcal{D}_\ell$ is negligible because of its small size. Specifically, the actual secrecy rate in the $\ell$-th level is given by $\frac{k_\ell}{k_\ell+1} [C(\Lambda_{\ell-1}/\Lambda_{\ell}, \sigma_b^2)-C(\Lambda_{\ell-1}/\Lambda_{\ell}, \sigma_e^2)]$. Clearly, this rate can be made close to the secrecy capacity by choosing sufficiently large $k_\ell$ as well.

\begin{ther}[Achieving secrecy capacity of the mod-$\Lambda_s$ GWC]\label{theorem:ratemodwiretap}
Consider a sequence of multi-level polar lattices $L(N)$ of increasing dimensions $N$. Let $L(N)$ be constructed according to \eqref{eqn:newassign} with the binary lattice partition chain $\Lambda/\Lambda_{1}/\cdot\cdot\cdot/\Lambda_r$ and $r$ binary nested polar codes where $r=O(\log N)$. Scale the lattice partition chain to satisfy the following conditions:
\begin{enumerate}[(i)]
  \item\label{ite:fir} $\epsilon_{\Lambda}(\sigma_b) \to 0$,
  \item\label{ite:sec} $\epsilon_e=\frac{1}{2}\log(2\pi e\sigma_e^2)-h(\Lambda_r, \sigma_e^2) \to 0$.
\end{enumerate}
Given $\sigma_e^2>\sigma_b^2$, the secrecy capacity $\frac{1}{2}\log\frac{\sigma_e^2}{\sigma_b^2}$ of the mod-$\Lambda_s$ Gaussian wiretap channel is achievable by using the polar lattices $L(N)$, i.e., for any rate $R<\frac{1}{2}\log\frac{\sigma_e^2}{\sigma_b^2}$, there exists a sufficiently large $N$ such that the realized rate $R(N)$ of $L(N)$ satisfies $R(N)>R$.
\end{ther}
\color{black}
\begin{proof}
By Lemma \ref{lem:securerate} and \eqref{eqn:newassign},
{\allowdisplaybreaks
\begin{eqnarray}\label{eq:achievable-rate}
\begin{aligned}
\lim_{N\rightarrow\infty} R(N)&=\sum_{\ell=1}^{r}\lim_{N\rightarrow\infty}\frac{|\mathcal{A}_\ell|}{N} \\
&=\sum_{\ell=1}^{r}C(V_\ell)-C(W_\ell) \\
&=\sum_{\ell=1}^{r}C(V(\Lambda_{\ell-1}/\Lambda_{\ell},\sigma_b^2))-C(W(\Lambda_{\ell-1}/\Lambda_{\ell},\sigma_e^2)) \\
&=C(V(\Lambda/\Lambda_{r},\sigma_b^2))-C(W(\Lambda/\Lambda_{r},\sigma_e^2))\\
&=C(\Lambda_r, \sigma_b^2)-C(\Lambda, \sigma_b^2)-C(\Lambda_r, \sigma_e^2)+C(\Lambda, \sigma_e^2) \\
&=h(\Lambda_r, \sigma_e^2)-h(\Lambda_r, \sigma_b^2)+h(\Lambda, \sigma_b^2)-h(\Lambda, \sigma_e^2) \\
&=\frac{1}{2}\log\frac{\sigma_e^2}{\sigma_b^2}-(\epsilon_e-\epsilon_b)-\epsilon_1,
\end{aligned}
\end{eqnarray}}
where
\begin{equation} \notag\
\begin{cases} \epsilon_1=C(\Lambda,\sigma_b^2)-C(\Lambda,\sigma_e^2)=h(\Lambda, \sigma_e^2)-h(\Lambda, \sigma_b^2)\geq0, \\
\epsilon_b=h(\sigma_b^2)-h(\Lambda_r, \sigma_b^2)=\frac{1}{2}\log(2\pi e\sigma_b^2)-h(\Lambda_r, \sigma_b^2)\geq0, \\
\epsilon_e=h(\sigma_e^2)-h(\Lambda_r, \sigma_e^2)=\frac{1}{2}\log(2\pi e\sigma_e^2)-h(\Lambda_r, \sigma_e^2)\geq0
\end{cases}
\end{equation}
and $\epsilon_e-\epsilon_b\geq0$.

By scaling $\Lambda$, we can have $h(\Lambda, \sigma_b^2) \rightarrow \log(\text{Vol}(\Lambda))$. Since $\sigma_e^2 > \sigma_b^2$, we also have $h(\Lambda, \sigma_e^2) \rightarrow \log(\text{Vol}(\Lambda))$. More precisely, by \cite[Lemma 1]{polarlatticeJ}, $\epsilon_1$ can be upper-bounded by the flatness factor as
\begin{eqnarray}
\epsilon_1 \leq C(\Lambda,\sigma_b^2) \leq \log(e)\cdot \epsilon_{\Lambda}(\sigma_b). \notag
\end{eqnarray}Then, according to \cite[Corollary 1]{cong2}, we can make $\epsilon_{\Lambda}(\sigma_b) \to 0$ by scaling $\Lambda$.
\color{black}

The number of levels is set such that $h(\Lambda_r, \sigma_e^2) \to \frac{1}{2}\log(2\pi e\sigma_e^2)$. By \cite[Theorem 2]{polarlatticeJ}, $r=O(\log N)$ is sufficient to guarantee $P_e(\Lambda_r, \sigma_b^2)=e^{-\Omega(N)}$, meaning that the volume $\text{Vol}(\Lambda_r)$ is sufficiently large such that $h(\Lambda_r, \sigma_e^2) \to \frac{1}{2}\log(2\pi e\sigma_e^2)$ as $N \to \infty$. Again, since $\sigma_e^2 > \sigma_b^2$, we immediately have $h(\Lambda_r, \sigma_b^2)\to\frac{1}{2}\log(2\pi e\sigma_e^2)$, and $\epsilon_e -\epsilon_b \to 0$. Therefore by scaling $\Lambda$ and adjusting $r$, the secrecy rate can get arbitrarily close to $\frac{1}{2}\log\frac{\sigma_e^2}{\sigma_b^2}$.
\end{proof}

\begin{rem}\label{rmk:sgood}
The constructed lattice $\Lambda_e$ is secrecy-good in the sense of Definition \ref{deft:secrecygood}. Recall that $\Lambda_e$ is constructed from the partition chain $\Lambda/\cdots/\Lambda_r$, which gives us the $N$-dimensional partition chain  $\Lambda^N/\Lambda_e/\Lambda_r^N$. Then,
\begin{eqnarray}
\begin{aligned}
C(\Lambda_e,\sigma_e^2)&=C(\Lambda^N,\sigma_e^2)+C(\Lambda^N/\Lambda_e,\sigma_e^2) \\\notag
&=C(\Lambda^N,\sigma_e^2)+I(\M\F; \Z^{[N]}) \\
&\leq \log(e)\cdot \epsilon_{\Lambda^N}(\sigma_e)+I(\M\F; \Z^{[N]}) \\
&\leq \log(e)\cdot ([1+\epsilon_{\Lambda}(\sigma_e)]^N-1)+I(\M\F; \Z^{[N]}),
\end{aligned}
\end{eqnarray}
where we use \cite[Corollary 1]{cong2} and \cite[Lemma 3]{LingBel13} in the last two inequalities, respectively.

Since $r=O (\log N)$, the top lattice $\Lambda$ can be scaled down so that $\epsilon_{\Lambda}(\sigma_e)$ vanishes as fast as $O(2^{-\sqrt{N}})$ by \cite[Proposition 2]{polarlatticeQZ}. When $N \to \infty$, we have
\begin{eqnarray}
C(\Lambda_e,\sigma_e^2)\leq N\log(e)\cdot \epsilon_{\Lambda}(\sigma_e) + I(\M\F; \Z^{[N]}) + O(2^{-\sqrt{N}}). \notag
\end{eqnarray}
Recalling \eqref{eqn:upperbound}, we immediately have $C(\Lambda_e, \sigma_e^2)  \to 0$. 

Meanwhile, following the analysis of \cite{polarlatticeJ}, we can show that the VNR $\gamma_{\Lambda_e}(\sigma_e^2)\to 2 \pi e$ from below. More precisely, the logarithmic VNR of $\Lambda_e$ satisfies
\begin{eqnarray}
\log\left(\frac{\gamma_{L}(\sigma)}{2\pi e}\right)
=2(\epsilon_{e1}-\epsilon_{e2}-\epsilon_{e3}) \notag\
\end{eqnarray}
where
\begin{equation}
\begin{cases} \epsilon_{e1}=C(\Lambda,\sigma_e^2) \\
\epsilon_{e2}=\frac{1}{2}\log2\pi e\sigma_e^2-h(\Lambda_r,\sigma_e^2) \\
\epsilon_{e3}=\sum_{\ell=1}^{r}{R_{\ell} - C(\Lambda_{\ell-1}/\Lambda_{\ell}, \sigma_e^{2})}.
\end{cases}
\label{eqn:epsilons}
\end{equation}
We note that, $\epsilon_{e1}\leq C(\Lambda,\sigma_b^2)\to 0$, $\epsilon_{e2}\to 0$ (condition (ii) in Theorem \ref{theorem:ratemodwiretap}), and $\epsilon_{3}$ is the total extra rate of component codes to guarantee security. Since $R_{\ell} = |\mathsf{R}_{\ell}|/N = (|\mathcal{B}_{r}|+|\mathcal{D}_{r}|)/N \to C(\Lambda_{\ell-1}/\Lambda_{\ell}, \sigma_e^{2})$, we also have $\epsilon_{3}\to 0$.

Let $U_{\mathcal{R}(\Lambda_e)}$ denote the uniform distribution over a fundamental region $\mathcal{R}(\Lambda_e)$. Note that condition $C(\Lambda_e, \sigma_e^2)  \to 0$ implies the following statements, which all state that the distribution $f_{\sigma_e,\Lambda_e}$ of the mod-$\Lambda_e$ Gaussian noise converges to the uniform distribution:
\begin{enumerate}
  \item Differential entropy $h(\Lambda_e,\sigma_e^2) \to \log (\mathrm{Vol}(\Lambda_e))$;
  \item Kullback-Leibler divergence $\D(f_{\sigma_e,\Lambda_e}\|U_{\mathcal{R}(\Lambda_e)}) \to 0$;
  \item Variational distance $\V(f_{\sigma_e,\Lambda_e},U_{\mathcal{R}(\Lambda_e)}) \to 0$
\end{enumerate}
where 1) is by definition, 2) from the relation between mutual information and Kullback-Leibler divergence\footnote{In fact, it is easy to show that $\D(f_{\sigma_e,\Lambda_e}\|U_{\mathcal{R}(\Lambda_e)}) = \log (\mathrm{Vol}(\Lambda_e)) - h(\Lambda_e,\sigma_e^2) = C(\Lambda_e, \sigma_e^2)$, thanks to the symmetry of the mod-$\Lambda_e$ channel.}, and 3) by Pinsker's inequality.

\end{rem}
\color{black}

\begin{rem}
The secrecy capacity of the mod-$\Lambda_s$ Gaussian wiretap channel per use is given by
\[
C_s = \frac{1}{N}C(\Lambda_s,\sigma_b^2) - \frac{1}{N}C(\Lambda_s,\sigma_e^2)=\frac{1}{N}h(\Lambda_s,\sigma_e^2) - \frac{1}{N} h(\Lambda_s,\sigma_b^2)
\]
since the wiretapper's channel is degraded with respect to the main channel. Because $h(\Lambda_r,\sigma_e^2)\rightarrow\frac{1}{2}\log(2\pi e\sigma_e^2)$ and $\Lambda_s \subset \Lambda_r^N$, we have $\frac{1}{N}h(\Lambda_s,\sigma_e^2) \to \frac{1}{2}\log(2\pi e\sigma_e^2)$ and $\frac{1}{N}h(\Lambda_s,\sigma_b^2) \to \frac{1}{2}\log(2\pi e\sigma_b^2)$. Hence $C_s \to \frac{1}{2}\log\frac{\sigma_e^2}{\sigma_b^2}$. It also equals the secrecy capacity of the Gaussian wiretap channel when the signal power goes to infinity. It is noteworthy that we successfully remove the $\frac{1}{2}$-nat gap in the achievable secrecy rate derived in \cite{cong2} which is caused by the limitation of the $L^{\infty}$ distance associated with the flatness factor.
\end{rem}

\begin{rem}
The mild conditions \eqref{ite:fir} and \eqref{ite:sec} stated in the theorem are easy to meet, by scaling top lattice $\Lambda$ and choosing the number of levels $r$ appropriately. Consider an example for $\sigma_e^2=4$ and $\sigma_b^2=1$. We choose $r=3$ levels and a partition chain $\mathbb{Z}/2\mathbb{Z}/4\mathbb{Z}$ with scaling factor $2.5$. The difference between the achievable rate computed from \eqref{eq:achievable-rate} and the upper bound $\frac{1}{2}\log\frac{\sigma_e^2}{\sigma_b^2}$ on secrecy capacity is about $0.05$.
\end{rem}

\begin{rem}\label{rmk:levelnum}
From conditions \eqref{ite:fir} and \eqref{ite:sec}, we can see that the construction for secrecy-good lattices requires more levels than the construction of AWGN-good lattices. $\epsilon_1$ can be made arbitrarily small by scaling down $\Lambda$ such that both $h(\Lambda, \sigma_e^2)$ and $h(\Lambda, \sigma_b^2)$ are sufficiently close to $\log (\text{Vol}(\Lambda))$. For polar lattices for AWGN-goodness \cite{yan2}, we only need $h(\Lambda_{r'}, \sigma_b^2)\approx\frac{1}{2}\log(2\pi e\sigma_b^2)$ for some $r'<r$. Since $\epsilon_b<\epsilon_e$, $\Lambda_{r'}$ may be not enough for the wiretapper's channel. Therefore, more levels are needed in the wiretap coding context. To satisfy the condition $h(\Lambda_r,\sigma_e^2)\rightarrow\frac{1}{2}\log(2\pi e\sigma_e^2)$, it is sufficient to guarantee that $P_e(\Lambda_r, \sigma_e^2)\rightarrow 0$ by \cite[Theorem 13]{forney6}. When one-dimensional binary partition $\mathbb{Z}/2\mathbb{Z}/4\mathbb{Z}/...$ is used, we have $P_e(\Lambda_r, \sigma_e^2) \leq Q(\frac{2^r}{2\sigma_e}) \leq e^{-\frac{2^{2r}}{8\sigma_e^2}}$, where $Q(\cdot)$ is the Q-function. Letting $r=O(\log N)$, the error probability vanishes as $P_e(\Lambda_r, \sigma_e^2)=e^{-\Omega(N)}$, which implies that $h(\Lambda_r,\sigma_e^2)\rightarrow\frac{1}{2}\log(2\pi e\sigma_e^2)$ as $N \rightarrow \infty$. We also note that when lattice Gaussian shaping is considered in Sect. \ref{sec:SecrecyGoodShap}, the probability of selecting a lattice point from $\Lambda_r$ decays exponentially as $r$ increases. The requirement is relaxed to $r=O(\log\log(N))$ to achieve the secrecy capacity.
\end{rem}

\subsection{Semantic security}

So far we have assumed that the message is uniformly distributed. In fact, this assumption is not needed because of the symmetry of the $\Lambda_b/\Lambda_e$ channel \cite{forney6}. It is well known that the error probability of polar codes in a symmetric channel is independent of the transmitted message \cite{arikan2009channel}; thus the input distribution does not matter for reliability. Moreover, the foregoing security analysis also implies \emph{semantic security}, i.e., \eqref{eqn:upperbound} holds for arbitrarily distributed $\M$ and $\F$. This $\Lambda_b/\Lambda_e$ channel can be seen as the counterpart in lattice coding of the randomness-induced channel defined in \cite{polarsecrecy}.

\begin{prop}\label{rmk:semtic}
Semantic security holds for the polar lattice construction for the mod-$\Lambda_s$ GWC shown in Fig. \ref{fig:eve}, i.e.,
\[
I\Big(\M \F;\Z^{[N]}\Big) \leq rN2^{-N^{\beta'}}
\]
for arbitrarily distributed $\M$ and $\F$.
\end{prop}

\begin{proof}
Since $\M\F$ is drawn from $\mathcal{R}(\Lambda_e)$ and the random bits are drawn from $\Lambda_e \cap \mathcal{R}(\Lambda_s)$, by Lemma \ref{lem:modsuffic}, the mod-$\Lambda_e$ map is information lossless and its output is a sufficient statistic for $\M\F$. Therefore, the channel between  $\M\F$ and the eavesdropper can be viewed as a $\Lambda_b/\Lambda_e$ channel. Because the $\Lambda_b/\Lambda_e$ channel is symmetric, the maximum mutual information is achieved by the uniform input. Consequently,
the mutual information corresponding to other input distributions can also be upper-bounded by $rN2^{-N^{\beta'}}$ as in \eqref{eqn:upperbound}, and we can also freeze the bits $\F$.
\end{proof}
\color{black}

\section{Achieving Secrecy Capacity with Discrete Gaussian Shaping}
\label{sec:SecrecyGoodShap}

In this section, we apply Gaussian shaping on the AWGN-good and secrecy-good polar lattices. The idea of lattice Gaussian shaping was proposed in \cite{LingBel13} and then implemented in \cite{polarlatticeJ} to construct capacity-achieving polar lattices. For wiretap coding, the discrete Gaussian distribution can also be utilized to satisfy the power constraint. In simple terms, after obtaining the AWGN-good lattice $\Lambda_b$ and the secrecy-good lattice $\Lambda_e$, Alice  maps each message $m$ to a coset $\widetilde{\lambda}_m \in \Lambda_b/\Lambda_e$ as mentioned in Sect. \ref{sec:NoPower}. However, instead of the mod-$\Lambda_s$ operation, Alice samples the encoded signal $\X^N$ from $D_{\Lambda_e+\lambda_m,\sigma_s}$, where $\lambda_m$ is the coset representative of $\widetilde{\lambda}_m$ and $\sigma_s^2$ is arbitrarily close to the signal power $P_s$ (see \cite{cong2} for more details). Again, we assume uniform messages until we prove semantic security in the end of this section.

The construction of polar lattices with Gaussian shaping is reviewed in Sect. \ref{sec:PLShape}. With Gaussian shaping, we propose a new partition of the index set for the genuine GWC in Sect. \ref{sec:3partition}. Strong secrecy is proved in Sect. \ref{sec:shapestrong}, and reliability is then discussed in Sect. \ref{sec:reliabilityshape}. Extension to semantical security is given in Sect. \ref{sec:semantic}. Moreover, we will show that this shaping operation does not hurt the secrecy rate and that the secrecy capacity can be achieved.

\subsection{Gaussian shaping over polar lattices}\label{sec:PLShape}
In this subsection, we introduce the lattice shaping technique for polar lattices. The idea is to select the lattice points according to a carefully chosen lattice Gaussian distribution, which makes a non-uniform input distribution for each partition channel.
\color{black}As shown in \cite{polarlatticeJ}, the shaping scheme is based on the technique of polar codes for asymmetric channels. For the paper to be self-contained, a brief review will be presented in this subsection. A more detailed account of Gaussian shaping can be found in \cite{polarlatticeJ}.

Similarly to the polar coding on symmetric channels, the Bhattacharyya parameter for a binary memoryless asymmetric (BMA) channel is defined as follows.
\begin{deft}[Bhattacharyya parameter for BMA channel]\label{deft:asymZ}
Let $W$ be a BMA channel with input $\X \in \mathcal{X}=\{0,1\}$ and output $\Y \in \mathcal{Y}$. The input distribution and channel transition probability is denoted by $P_\X$ and $P_{\Y|\X}$ respectively. The Bhattacharyya parameter $Z$ for $W$ is the defined as
\begin{eqnarray}
Z(\X|\Y)&=&2\sum\limits_{y}P_\Y(y)\sqrt{P_{\X|\Y}(0|y)P_{\X|\Y}(1|y)} \notag\ \\
&=&2\sum\limits_{y}\sqrt{P_{\X,\Y}(0,y)P_{\X,\Y}(1,y)}. \notag
\end{eqnarray}
\end{deft}

 The following lemma, which will be useful for the forthcoming new partition scheme, shows that by adding observable at the output of $W$, $Z$ will not increase. \color{black}
\begin{lem}[Conditioning reduces Bhattacharyya parameter $Z$ \cite{polarlatticeJ}]\label{lem:CondiReduce}
Let $(\X,\Y,\Y')\sim P_{\X,\Y,\Y'}, X\in\mathcal{\X}=\{0,1\}, \Y\in \mathcal{Y},\Y'\in \mathcal{Y}'$, we have
\begin{eqnarray}
Z(\X|\Y,\Y')\leq Z(\X|\Y). \notag
\end{eqnarray}
\end{lem}

When $\X$ is uniformly distributed, the Bhattacharyya parameter of BMA channels coincides with that of BMS channels defined in Definition \ref{deft:symmZ}. Moreover, the calculation of $Z$ can be converted to the calculation of the Bhattacharyya parameter $\widetilde{Z}$ for a related BMS channel. The following lemma is implicitly considered in \cite{aspolarcodes} and then explicitly expressed in \cite{polarlatticeJ}. We show it here for completeness.

\begin{lem}[From Asymmetric to Symmetric channel \cite{polarlatticeJ}]\label{lem:Asy2sym}
Let $W$ be a binary input asymmetric channel with input $\X \in \mathcal{X}=\{0,1\}$ and $\Y \in \mathcal{Y}$.
We define a new channel $\widetilde{W}$ corresponding to $W$ which has input $\widetilde{\X} \in \mathcal{X}=\{0,1\}$ and output $\widetilde{\Y} \in \mathcal{Y} \times \mathcal{X}$. \color{black}The relationship between $\widetilde{W}$ and $W$ is shown in Fig.  \ref{fig:2Asym2sym}. The input of $\widetilde{W}$ is uniformly distributed, i.e., $P_{\widetilde{\X}} (\widetilde{x}=0)=P_{\widetilde{\X}} (\widetilde{x}=1)=\frac{1}{2}$, and the output of $\widetilde{W}$ is given by $(\Y, \X \oplus \widetilde{\X})$, where $\oplus$ denotes the bitwise XOR operation. Then, $\widetilde{W}$ is a binary symmetric channel in the sense that $P_{\widetilde{\Y}|\widetilde{\X}}(y,x\oplus\widetilde{x}|\widetilde{x})=P_{\Y,\X} (y,x)$.
\end{lem}

\begin{figure}[h]
    \centering
    \includegraphics[width=7cm]{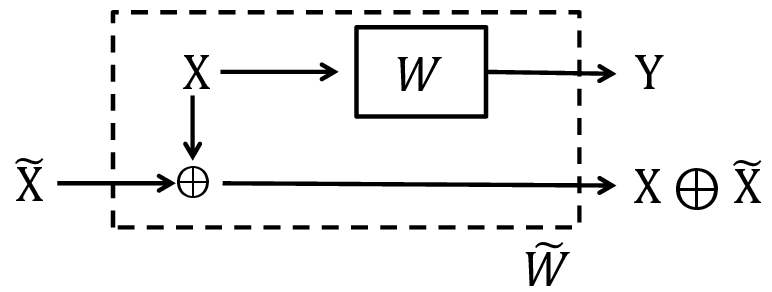}
    \caption{The relationship between $\widetilde{W}$ and $W$.}
    \label{fig:2Asym2sym}
\end{figure}

\color{black}

The following lemma describes how to construct a polar code for a BMA channel $W$ from that for the associated BMS channel $\widetilde{W}$.
\begin{lem}[The equivalence between symmetric and asymmetric Bhattacharyya parameters \cite{aspolarcodes}]\label{lem:conasypolar}
For a BMA channel $W$ with input $\X \sim P_{\X}$, let $\widetilde{W}$ be its symmetrized channel constructed according to Lemma \ref{lem:Asy2sym}. Suppose $\X^{[N]}$ and $\Y^{[N]}$ be the input and output vectors of $W^N$, and let $\widetilde{\X}^{[N]}$ and $\widetilde{\Y}^{[N]}=\left(\X^{[N]}\oplus\widetilde{\X}^{[N]},\Y^{[N]}\right)$ be the input and output vectors of $\widetilde{W}^N$, where $\widetilde{\X}$ is uniform. Consider polarized random variables $\U^{[N]}$=$\X^{[N]}G_N$ and $\widetilde{\U}^{[N]}$=$\widetilde{\X}^{[N]}G_N$, and denote by $W_N$ and $\widetilde{W}_N$ the combining channel of $N$ uses of $W$ and $\widetilde{W}$, respectively. The Bhattacharyya parameter for each subchannel of $W_N$ is equal to that of each subchannel of $\widetilde{W}_N$, i.e.,
\begin{eqnarray}
Z\Big(\U^i|\U^{1:i-1},\Y^{[N]}\Big)=\widetilde{Z}\Big(\widetilde{\U}^i|\widetilde{\U}^{1:i-1},\X^{[N]}\oplus\widetilde{\X}^{[N]},\Y^{[N]}\Big). \notag\
\label{eqn:asymmetricz}
\end{eqnarray}
\end{lem}

To obtain the desired input distribution of $P_\X$ for $W$, the indices with very small $Z(\U^i|\U^{1:i-1})$ should be removed from the information set of the symmetric channel. Following \cite{polarlatticeJ}, the resultant subset is referred to as the information set $\mathcal{I}$ for the asymmetric channel $W$. For the remaining part $\mathcal{I}^c$, we further find out that there are some bits which can be made independent of the information bits and uniformly distributed. The purpose of extracting such bits is for the interest of our lattice construction. We name the set that includes those independent frozen bits as the independent frozen set $\mathcal{F}$, and the remaining frozen bits are determined by the bits in $\mathcal{F}\cup\mathcal{I}$. We name the set of all those deterministic bits as the shaping set $\mathcal{S}$. The three sets are formally defined as follows:
\begin{eqnarray}\label{eqn:asymdefinition}
&&\hspace{-2em}
\begin{cases}
\begin{aligned}
&\text{the independent frozen set: } \mathcal{F}=\Big\{i\in[N]:Z(\U^i|\U^{1:i-1},\Y^{[N]})\geq1-2^{-N^{\beta}}\Big\}\\
&\text{the information set: } \mathcal{I}=\Big\{i\in[N]:Z(\U^i|\U^{1:i-1},\Y^{[N]})\leq2^{-N^{\beta}}\text{ and }Z(\U^i|\U^{1:i-1})\geq1-2^{-N^{\beta}}\Big\}\\
&\text{the shaping set: } \mathcal{S}=\left(\mathcal{F}\cup\mathcal{I}\right)^c.
\end{aligned}
\end{cases}
\end{eqnarray}

To identify these three sets, one can use Lemma \ref{lem:conasypolar} to calculate $Z(\U^i|\U^{1:i-1},\Y^{[N]},\X^{[N]})$ using the known constructing techniques for symmetric polar codes \cite{Ido}\cite{mori2009performance}. We note that $Z(\U^i|\U^{1:i-1})$ can be computed in a similar way, by constructing a symmetric channel between $\widetilde{\X}$ and $\X \oplus \widetilde{\X}$. Besides the construction, the decoding process for the asymmetric polar codes can also be converted to the decoding for the symmetric polar codes.

The polar coding scheme according to \eqref{eqn:asymdefinition}, which can be viewed as an extension of the scheme proposed in \cite{aspolarcodes}, has been proved to be capacity-achieving in \cite{polarlatticeJ}. Moreover, it can be extended to the construction of multilevel asymmetric polar codes.

Let us describe the encoding strategy for the channel of the $\ell$-th ($\ell \leq r$) level $W_\ell$ with the channel transition probability $P_{\Y|\X_\ell,\X_{1:\ell-1}}(y|x_\ell,x_{1:\ell-1})$ as follows.

\begin{itemize}
\item Encoding: Before sending the codeword $x_\ell^{[N]}=u_\ell^{[N]}G_N$, the index set $[N]$ are divided into three parts: the independent frozen set $\mathcal{F}_\ell$, information set $\mathcal{I}_\ell$, and shaping set $\mathcal{S}_\ell$, which are defined as follows:
    \begin{eqnarray}\label{eqn:asymdefinition1}
    &&\notag\
    \begin{cases}
    \begin{aligned}
    & \mathcal{F}_\ell=\Big\{i\in[N]:Z\Big(\U_\ell^i|\U_\ell^{1:i-1},\X_{1:\ell-1}^{[N]},\Y^{[N]}\Big)\geq1-2^{-N^{\beta}}\Big\}\\
    & \mathcal{I}_\ell=\Big\{i\in[N]:Z\Big(\U_\ell^i|\U_\ell^{1:i-1},\X_{1:\ell-1}^{[N]},\Y^{[N]}\Big)\leq2^{-N^{\beta}}\text{ and }Z\Big(\U_\ell^i|\U_\ell^{1:i-1}, \X_{1:\ell-1}^{[N]}\Big)\geq1-2^{-N^{\beta}}\Big\}\\
    & \mathcal{S}_\ell=\left(\mathcal{F}_\ell\cup\mathcal{I}_\ell\right)^c.
    \end{aligned}
    \end{cases}
    \end{eqnarray}

The encoder first places uniformly distributed information bits in $\mathcal{I}_\ell$. Then the frozen set $\mathcal{F}_\ell$ is filled with a uniform random sequence which is shared between the encoder and the decoder. The bits in $\mathcal{S}_\ell$ are generated by a random mapping $\Phi_{\mathcal{S}_\ell}$, which yields the following distribution:
\begin{equation}
u_\ell^i=
\begin{cases}
0 \;\;\;\; \text{with probability } P_{\U_\ell^i|\U_\ell^{1:i-1},\X_{1:\ell-1}^{[N]}}(0|u_\ell^{1:i-1},x_{1:\ell-1}^{[N]}),\\
1 \;\;\;\; \text{with probability } P_{\U_\ell^i|\U_\ell^{1:i-1},\X_{1:\ell-1}^{[N]}}(1|u_\ell^{1:i-1},x_{1:\ell-1}^{[N]}).
\end{cases}
\label{eqn:encodingside}
\end{equation}
\end{itemize}

\begin{ther}[Construction of multilevel polar codes \cite{polarlatticeJ}]\label{theorem:codingtheoremside}
Consider a polar code with the above encoding strategy. Then, any message rate arbitrarily close to $I(\X_\ell;\Y|\X_{1:\ell-1})$ is achievable using SC
decoding\footnote{It is possible to derandomize the mapping $\Phi_{\mathcal{S}_\ell}$ for the purpose of achieving capacity alone. However, it is tricky to handle the random mapping in order to achieve the secrecy capacity: it requires either to share a secret random mapping or to use the Markov block coding technique (see Sect. \ref{sec:reliabilityshape}).} and the expectation of the decoding error probability over the randomized mappings satisfies $E_{\Phi_{\mathcal{S}_\ell}}[P_e(\phi_{\mathcal{S}_\ell})]=O(2^{-N^{\beta'}})$ for any $\beta'<\beta<0.5$.
\end{ther}

\color{black}

Now let us pick a suitable input distribution $P_{\X_{1:r}}$ to implement the shaping. As shown in Theorem \ref{theorem:capacity}, the mutual information between the discrete Gaussian lattice distribution $D_{\Lambda,\sigma_s}$ and the output of the AWGN channel approaches $\frac{1}{2}\log(1+\SNR)$ as the flatness factor $\epsilon_{\Lambda}(\widetilde{\sigma})\rightarrow0$. Therefore, we use the lattice Gaussian distribution $P_{\X}\sim D_{\Lambda,\sigma_s}$ as the constellation, which gives us $\lim_{r\rightarrow\infty}P_{\X_{1:r}}=P_{\X}\sim D_{\Lambda,\sigma_s}$. By \cite[Lemma 5]{polarlatticeJ}, when $N \rightarrow \infty$, the mutual information $I(\X_{r};\Y|\X_{1:r-1})$ at the bottom level goes to 0 if $r=O(\log\log N)$, and using the first $r$ levels would involve a capacity loss $\sum_{\ell > r} I(\X_{\ell};\Y|\X_{1:\ell-1})\leq O(\frac{1}{N})$.

From the chain rule of mutual information,
\begin{eqnarray}\label{eqn:chainrule}
I(\X_{1:r};\Y)=\sum_{\ell=1}^{r} I(\X_\ell;\Y|\X_{1:\ell-1}), \notag
\end{eqnarray}
we have $r$ binary-input channels and the $\ell$-th channel according to $I(\X_\ell;\Y|\X_{1:\ell-1})$ is generally asymmetric with the input distribution $P_{\X_{\ell}|\X_{1:\ell-1}}$ $(1\leq \ell\leq r)$. Then we can construct the polar code for the asymmetric channel at each level according to Lemma \ref{lem:Asy2sym}. As a result, the $\ell$-th symmetrized channel is equivalent to the MMSE-scaled $\Lambda_{\ell-1}/\Lambda_{\ell}$ channel in the sense of channel polarization. (See \cite{polarlatticeJ} for more details.)

Therefore, when power constraint is taken into consideration, the multilevel polar codes before shaping are constructed according to the symmetric channel $V(\Lambda_{\ell-1}/\Lambda_{\ell}, \widetilde{\sigma}_b^2)$ and $W(\Lambda_{\ell-1}/\Lambda_{\ell}, \widetilde{\sigma}_e^2)$, where $\widetilde{\sigma}_b^2=\Big(\frac{\sigma_s\sigma_b}{\sqrt{\sigma_s^2+\sigma_b^2}}\Big)^2$ and $\widetilde{\sigma}_e^2=\Big(\frac{\sigma_s\sigma_e}{\sqrt{\sigma_s^2+\sigma_e^2}}\Big)^2$ are the MMSE-scaled noise variance of the main channel and of the wiretapper's channel, respectively. This is similar to the mod-$\Lambda_s$ GWC scenario mentioned in the previous section. The difference is that $\sigma_b^2$ and $\sigma_e^2$ are replaced by $\widetilde{\sigma}_b^2$ and $\widetilde{\sigma}_e^2$ accorrdingly. As a result, we can still obtain an AWGN-good lattice $\Lambda_b$ and a secrecy-good lattice $\Lambda_e$ by treating $V(\Lambda_{\ell-1}/\Lambda_{\ell}, \widetilde{\sigma}_b^2)$ and $W(\Lambda_{\ell-1}/\Lambda_{\ell}, \widetilde{\sigma}_e^2)$ as the main channel and wiretapper's channel at each level.

\subsection{Three-dimensional partition}\label{sec:3partition}
When lattice Gaussian shaping is performed over the AWGN-good lattice $\Lambda_b$ and the secrecy-good lattice $\Lambda_e$ simultaneously, we have a new shaping induced partition. The polar coding scheme for the mod-$\Lambda_s$ wiretap channel given in Sect. IV needs to be modified.
\color{black}Now we consider the partition of the index set $[N]$ with shaping involved. According to the analysis of asymmetric polar codes, we have to eliminate those indices with small $Z(\U_{\ell}^i|\U_{\ell}^{1:i-1},\X_{1:l-1}^{[N]})$ from the information set of the symmetric channels. Therefore, Alice cannot send message on those subchannels with $Z(\U_{\ell}^i|\U_{\ell}^{1:i-1},\X_{1:\ell-1}^{[N]}) < 1-2^{-N^\beta}$. Note that this part is the same for $\widetilde{V}_{\ell}$ and $\widetilde{W}_{\ell}$, because it only depends on the shaping distribution. At each level, the index set which is used for shaping is
given as
\begin{eqnarray}
\mathcal{S}_\ell \triangleq \Big\{i \in [N]: Z(\U_{\ell}^i|\U_{\ell}^{1:i-1},\X_{1:\ell-1}^{[N]}) < 1-2^{-N^\beta}\Big\},\notag\
 \notag\
\end{eqnarray}
and the index set which is not for shaping is denoted by $\mathcal{S}_\ell^c$. Recall that for the index set $[N]$, we already have two partition criteria, i.e, reliability-good and information-bad (see \eqref{eqn:Good&bad}). We rewrite the reliability-good index set $\mathcal{G}_\ell$ and information-poor index set $\mathcal{N}_\ell$ at level $\ell$ as
\begin{equation}\label{eqn:Good&bad2}
\begin{split}
\mathcal{G}_\ell &\triangleq \Big\{i \in [N]: Z(\U_{\ell}^i|\U_{\ell}^{1:i-1},\X_{1:\ell-1}^{[N]},\Y^{[N]}) \leq 2^{-N^\beta}\Big\}, \\
\mathcal{N}_\ell &\triangleq \Big\{i \in [N]: Z(\U_{\ell}^i|\U_{\ell}^{1:i-1},\X_{1:\ell-1}^{[N]},\Z^{[N]}) \geq 1-2^{-N^\beta}\Big\}.
\end{split}
\end{equation}

Note that $\mathcal{G}_\ell$ and $\mathcal{N}_\ell$ are defined by the asymmetric Bhattacharyya parameters. Nevertheless, by Lemma \ref{lem:conasypolar} and the channel equivalence, we have $\mathcal{G}_\ell=\mathcal{G}(\widetilde{V}_\ell)$ and $\mathcal{N}_\ell=\mathcal{N}(\widetilde{W}_\ell)$ as defined in \eqref{eqn:Good&bad}, where $\widetilde{V}_\ell$ and $\widetilde{W}_\ell$ are the respective symmetric channels or the MMSE-scaled $\Lambda_{\ell-1}/\Lambda_{\ell}$ channels for Bob and Eve at level $\ell$. The four sets $\mathcal{A}_\ell$, $\mathcal{B}_\ell$, $\mathcal{C}_\ell$, and $\mathcal{D}_\ell$ are defined in the same fashion as \eqref{eqn:partition}, with $\mathcal{G}_\ell$ and $\mathcal{N}_\ell$ replacing $\mathcal{G}(\widetilde{V}_\ell)$ and $\mathcal{N}(\widetilde{W}_\ell)$, respectively. Now the whole index set $[N]$ is divided like a cube in three directions, which is shown in Fig. \ref{fig:3parts}.
\begin{figure}[h]
    \centering
    \includegraphics[width=10cm]{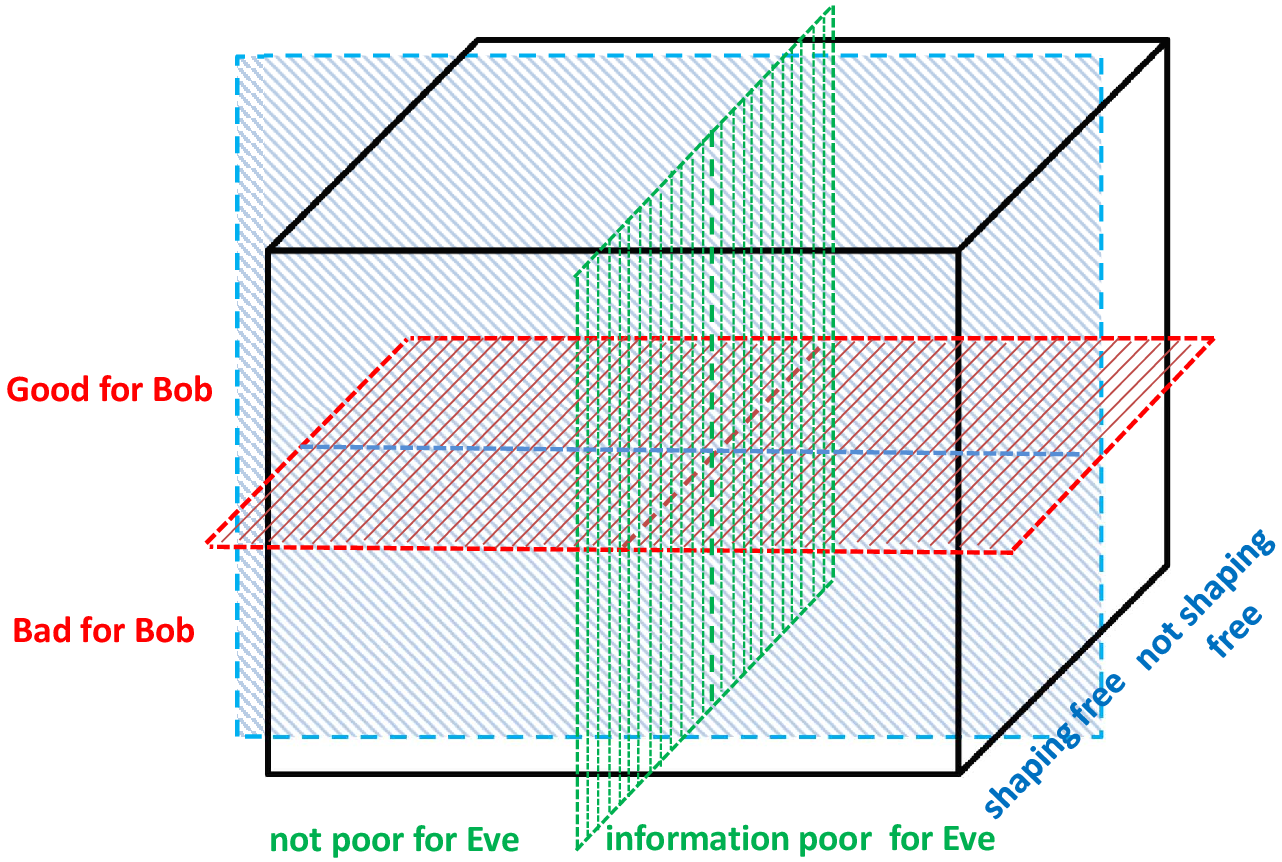}
    \caption{Partitions of the index set $[N]$ with shaping.}
    \label{fig:3parts}
\end{figure}

Clearly, we have eight blocks:
\begin{eqnarray}\label{eqn:shapingassign}
\begin{aligned}
&\mathcal{A}_\ell^{\mathcal{S}}=\mathcal{A}_\ell \cap \mathcal{S}_\ell, \; \mathcal{A}_\ell^{\mathcal{S}^c}=\mathcal{A}_\ell \cap \mathcal{S}_\ell^c\\
&\mathcal{B}_\ell^{\mathcal{S}}=\mathcal{B}_\ell \cap \mathcal{S}_\ell, \; \mathcal{B}_\ell^{\mathcal{S}^c}=\mathcal{B}_\ell \cap \mathcal{S}_\ell^c\\
&\mathcal{C}_\ell^{\mathcal{S}}=\mathcal{C}_\ell \cap \mathcal{S}_\ell, \; \mathcal{C}_\ell^{\mathcal{S}^c}=\mathcal{C}_\ell \cap \mathcal{S}_\ell^c\\
&\mathcal{D}_\ell^{\mathcal{S}}=\mathcal{D}_\ell \cap \mathcal{S}_\ell, \; \mathcal{D}_\ell^{\mathcal{S}^c}=\mathcal{D}_\ell \cap \mathcal{S}_\ell^c\\
\end{aligned}
\end{eqnarray}

By Lemma \ref{lem:CondiReduce}, we observe that $\mathcal{A}_\ell^{\mathcal{S}}=\mathcal{C}_\ell^{\mathcal{S}}=\emptyset$, $\mathcal{A}_\ell^{\mathcal{S}^c}=\mathcal{A}_\ell$, and $\mathcal{C}_\ell^{\mathcal{S}^c}=\mathcal{C}_\ell$. The shaping set $\mathcal{S}_\ell$ is divided into two sets $\mathcal{B}_\ell^{\mathcal{S}}$ and $\mathcal{D}_\ell^{\mathcal{S}}$. The bits in $\mathcal{S}_\ell$ are determined by the bits in $\mathcal{S}_\ell^c$ according to the mapping. Similarly, $\mathcal{S}_\ell^c$ is divided into the four sets $\mathcal{A}_\ell^{\mathcal{S}^c}=\mathcal{A}_\ell$, $\mathcal{B}_\ell^{\mathcal{S}^c}$, $\mathcal{C}_\ell^{\mathcal{S}^c}=\mathcal{C}_\ell$, and $\mathcal{D}_\ell^{\mathcal{S}^c}$. Note that for wiretap coding, the frozen set becomes $\mathcal{C}_\ell^{\mathcal{S}^c}$, which is slightly different from the frozen set for channel coding. To satisfy the reliability condition, the frozen set $\mathcal{C}_\ell^{\mathcal{S}^c}$ and the problematic set $\mathcal{D}_\ell^{\mathcal{S}^c}$ cannot be set uniformly random any more. Recall that only the independent frozen set $\mathcal{F}_\ell$ at each level, which is defined as $\{i\in[N]:Z(\U_\ell^i|\U_\ell^{1:i-1},\Y^{[N]},\X_{1:\ell-1}^{[N]})\geq1-2^{-N^{\beta}}\}$, can be set uniformly random (which are already shared between Alice and Bob), and the bits in the unpolarized frozen set $\bar{\mathcal{F}}_\ell$, defined as $\{i\in[N]:2^{-N^{\beta}}<Z(\U_\ell^i|\U_\ell^{1:i-1},\Y^{[N]},\X_{1:\ell-1}^{[N]})<1-2^{-N^{\beta}}\}$, should be determined according to the mapping. Moreover, we can observe that $\mathcal{F}_\ell \subset \mathcal{C}_\ell^{\mathcal{S}^c}$ and $\mathcal{D}_\ell^{\mathcal{S}^c}\subset \mathcal{D}_\ell \subset \bar{\mathcal{F}}_\ell$. Here we make the bits in $\mathcal{F}_\ell$ uniformly random and the bits in $\mathcal{C}_\ell^{\mathcal{S}^c}\setminus\mathcal{F}_\ell$ and $\mathcal{D}_\ell^{\mathcal{S}^c}$ determined by the mapping. Therefore, from now on, we adjust the definition of the shaping bits as:
\begin{eqnarray}\label{eqn:newS}
\mathcal{S}_\ell \triangleq \Big\{i \in [N]: Z(\U_{\ell}^i|\U_{\ell}^{1:i-1},\X_{1:\ell-1}^{[N]}) < 1-2^{-N^\beta} \,\text{or}\,\, 2^{-N^{\beta}}<Z(\U_\ell^i|\U_\ell^{1:i-1},\Y^{[N]},\X_{1:\ell-1}^{[N]})<1-2^{-N^{\beta}}\Big\},
\end{eqnarray}
which is essentially equivalent to the definition of the shaping set given in Theorem \ref{theorem:codingtheoremside}.

To sum up, at level $\ell$, we assign the sets $\mathcal{A}_\ell^{\mathcal{S}^c}$, $\mathcal{B}_\ell^{\mathcal{S}^c}$, and $\mathcal{F}_\ell$ with message bits $\M_\ell$, uniformly random bits $\mathsf{R}_\ell$, and uniform frozen bits $\F_\ell$, respectively. The rest bits $\mathsf{S}_\ell$ (in $\mathcal{S}_\ell$) will be fed with random bits according to $P_{\U_{\ell}^i|\U_{\ell}^{1:i-1},\X_{1:l-1}^{[N]}}$. Clearly, this shaping operation will make the input distribution arbitrarily close to $P_{\X_\ell|\X_{1:\ell-1}}$, for $\beta$ fixed and $N$ tending to infinity. In this case, we can obtain the equality between the Bhattacharyya parameter of asymmetric setting and symmetric setting (see Lemma \ref{lem:conasypolar}). This provides us a convenient way to prove the strong secrecy of the wiretap coding scheme with shaping because we have already proved the strong secrecy of a symmetric wiretap coding scheme using the Bhattacharyya parameter of the symmetric setting. A detailed proof will be presented in the following subsection. Before this, we show that the shaping will not change the message rate.

\begin{lem}\label{lem:shapdoesnotmatter}
For the symmetrized main channel $\widetilde{V}_\ell$ and wiretapper's channel $\widetilde{W}_\ell$, consider the reliability-good indices set $\mathcal{G}_\ell$ and information-bad indices set $\mathcal{N}_\ell$ defined as in \eqref{eqn:Good&bad2}. By eliminating the shaping set $\mathcal{S}_\ell$ from the original message set defined in \eqref{eqn:partition}, we get the new message set $\mathcal{A}_\ell^{\mathcal{S}^c}=\mathcal{G}_\ell\cap \mathcal{N}_\ell\cap \mathcal{S}^c_\ell$. The proportion of $|\mathcal{A}_\ell^{\mathcal{S}^c}|$ equals to that of $|\mathcal{A}_\ell|$, and the message rate after shaping can still be arbitrarily close to $\frac{1}{2}\log\frac{\widetilde{\sigma}_e^2}{\widetilde{\sigma}_b^2}$.
\end{lem}

\begin{proof}
By Theorem \ref{theorem:ratemodwiretap}, when shaping is not involved, the message rate can be made arbitrarily close to $\frac{1}{2}\log\frac{\widetilde{\sigma}_e^2}{\widetilde{\sigma}_b^2}$. By the new definition \eqref{eqn:newS} of $\mathcal{S}_\ell$, we still have $\mathcal{A}_\ell^{\mathcal{S}}=\emptyset$, which means the shaping operation will not affect the message rate.
\end{proof}

\subsection{Strong secrecy}\label{sec:shapestrong}
In this subsection, we prove that strong secrecy can still be achieved when shaping is involved. To this end, we introduce a new induced channel from Eve's perspective and prove that the information leakage over this channel vanishes at each level. \color{black}
Then, strong secrecy is proved by using the chain rule of mutual information as in \eqref{eqn:upperbound}.

\begin{figure}[h]
    \centering
    \includegraphics[width=10cm]{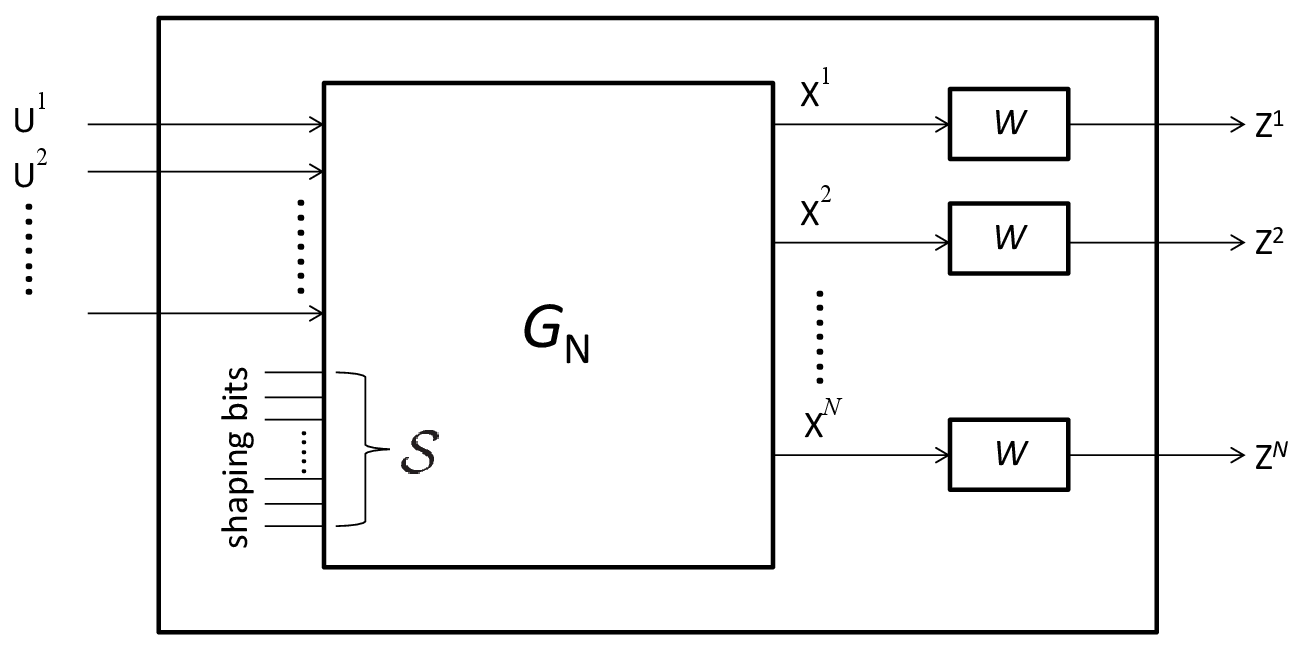}
    \caption{Block diagram of the shaping-induced channel $\mathcal{Q}_N(W,\mathcal{S})$.}
    \label{fig:shapingindu}
\end{figure}

In \cite{polarsecrecy}, an induced channel is defined in order to prove strong secrecy. Here we call it the randomness-induced channel because it is induced by feeding the subchannels in the sets $\mathcal{B}_\ell$ and $\mathcal{D}_\ell$ with uniformly random bits. However, when shaping is involved, the set $\mathcal{B}_\ell$ and $\mathcal{D}_\ell$ are no longer fed with uniformly random bits. In fact, some subchannels (covered by the shaping mapping) should be fed with bits according to a random mapping. We define the channel induced by the shaping bits as the shaping-induced channel.

\begin{deft}[Shaping-induced channel]
The shaping-induced channel $\mathcal{Q}_N(W,\mathcal{S})$ is defined in terms of $N$ uses of an asymmetric channel $W$, and a shaping subset $\mathcal{S}$ of $[N]$ of size $|\mathcal{S}|$. The input alphabet of $\mathcal{Q}_N(W,\mathcal{S})$ is $\{0,1\}^{N-|\mathcal{S}|}$ and the bits in $\mathcal{S}$ are determined by the input bits according to a random shaping $\Phi_{\mathcal{S}}$. A block diagram of the shaping induced channel is shown in Fig. \ref{fig:shapingindu}.
\end{deft}

Based on the shaping-induced channel, we define a new induced channel, which is caused by feeding a part of the input bits of the shaping-induced channel with uniformly random bits.

\begin{deft}[New induced channel]
Based on a shaping induced channel $\mathcal{Q}_N(W,\mathcal{S})$, the new induced channel $\mathcal{Q}_N(W,\mathcal{S}, \mathcal{R})$ is specified in terms of a randomness subset $\mathcal{R}$ of size $|\mathcal{R}|$. The randomness is introduced into the input set of the shaping-induced channel. The input alphabet of $\mathcal{Q}_N(W,\mathcal{S}, \mathcal{R})$ is $\{0,1\}^{N-|\mathcal{S}|-|\mathcal{R}|}$ and the bits in $\mathcal{R}$ are uniformly and independently random. A block diagram of the new induced channel is shown in Fig. \ref{fig:randomindu}.
\end{deft}

\begin{figure}[h]
    \centering
    \includegraphics[width=10cm]{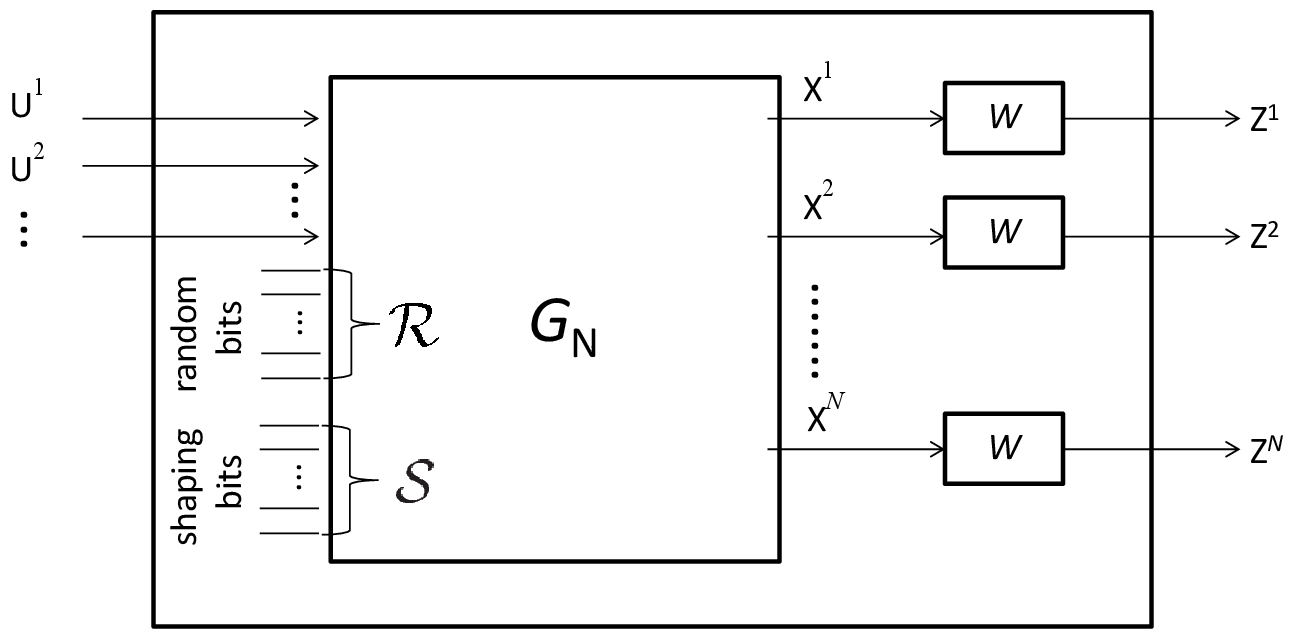}
    \caption{Block diagram of the new induced channel $\mathcal{Q}_N(W,\mathcal{S},\mathcal{R})$.}
    \label{fig:randomindu}
\end{figure}

The new induced channel is a combination of the shaping-induced channel and randomness-induced channel. This is different from the definition given in \cite{polarsecrecy} because the bits in $\mathcal{S}$ are neither independent to the message bits nor uniformly distributed. As long as the input bits of the new induced channel are uniform and the shaping bits are chosen according to the random mapping, the new induced channel can still generate $2^N$ possible realizations $x_\ell^{[N]}$ of $\X_\ell^{[N]}$ as $N$ goes to infinity, and those $x_\ell^{[N]}$ can be viewed as the output of $N$ i.i.d binary sources with input distribution $P_{\X_\ell|\X_{1:\ell-1}}$. These are exactly the conditions required by Lemma \ref{lem:conasypolar}. Specifically, we have $Z\Big(\U_\ell^i|\U_\ell^{1:i-1},\X_{1:\ell-1}^{[N]},\Z^{[N]}\Big)=\widetilde{Z}\Big(\widetilde{\U}_\ell^i|\widetilde{\U}_\ell^{1:i-1},\X_{1:\ell-1}^{[N]},\X_\ell^{[N]}\oplus\widetilde{\X}_\ell^{[N]},\Z^{[N]}\Big)$. In simple words, this equation holds when $x_\ell^{[N]}$ and $x_\ell^{[N]}\oplus\widetilde{x}_\ell^{[N]}$ are all selected from $\{0,1\}^N$ according to their respective distributions. Then we can exploit the relation between the asymmetric channel and the corresponding symmetric channel to bound the mutual information of the asymmetric channel. Therefore, we have to stick to the input distribution (uniform) of our new induced channel and also the distribution of the random mapping. This is similar to the setting of the randomness induced channel in \cite{polarsecrecy}, where the input distribution and the randomness distribution are both set to be uniform. In \cite{polarsecrecy}, the randomness-induced channel is further proved to be symmetric; then any other input distribution can also achieve strong secrecy and the symmetry finally results in semantic security. In this work, however, we do not have a proof of the symmetry of the new induced channel. For this reason, we assume for now that the message bits are uniform distributed. To prove semantic security, we will show that the information leakage of the symmetrized version of the new induced channel is vanishing in Sect. \ref{sec:semantic}.

\begin{lem}\label{lem:secrecybound}
Let $\M_\ell$ be the uniformly distributed message bits and $\F_\ell$ be the independent frozen bits at the input of the channel at the $\ell$-th level. When shaping bits $\mathsf{S}_\ell$ are selected according to the random mapping $\Phi_{\mathcal{S}_\ell}$ \footnote{We will further show that the number of shaping bits $\mathsf{S}_\ell$ covered by random mapping can be significantly reduced in Sect. V-E. Then, to achieve reliability, $\mathsf{S}_\ell$ can be shared between Alice and Bob, or we can use the Markov block coding technique to hide $\mathsf{S}_\ell$ with negligible rate loss.} and $N$ is sufficiently large, the mutual information can be upper-bounded as
\begin{eqnarray}
I\Big(\M_\ell \F_\ell;\Z^{[N]},\X_{1:\ell-1}^{[N]}\Big)\leq O(N^2 2^{-N^ {\beta'}}).\notag\
\end{eqnarray}
\end{lem}
\begin{proof}
We firstly assume that $\U_\ell^{i}$ is selected according to the distribution $P_{\U_\ell^i|\U_\ell^{1:i-1},\X_{1:\ell-1}^{[N]}}$ for all $i \in [N]$, i.e.,
\begin{equation}
u_\ell^i=
\begin{cases}
0 \;\;\;\; \text{with probability } P_{\U_\ell^i|\U_\ell^{1:i-1},\X_{1:\ell-1}^{[N]}}(0|u_\ell^{1:i-1},x_{1:\ell-1}^{[N]}),\\
1 \;\;\;\; \text{with probability } P_{\U_\ell^i|\U_\ell^{1:i-1},\X_{1:\ell-1}^{[N]}}(1|u_\ell^{1:i-1},x_{1:\ell-1}^{[N]}).
\end{cases}
\label{eqn:rdnmappingS}
\end{equation} for all $i \in [N]$. In this case, the input distribution $P_{\X_\ell|\X_{1:\ell-1}}$ at each level is exactly the optimal input distribution obtained from the lattice Gaussian distribution. The mutual information between $\M_\ell \F_\ell$ and
$\Big(\Z^{[N]},\X_{1:\ell-1}^{[N]}\Big)$ in this case is denoted by $I_P\Big(\M_\ell \F_\ell;\Z^{[N]},\X_{1:\ell-1}^{[N]}\Big)$.

For the shaping induced channel $\mathcal{Q}_N(W_\ell,\mathcal{S}_\ell,\mathcal{R}_\ell)$ ($\mathcal{R}_\ell$ is $\mathcal{B}_\ell^{\mathcal{S}^c}$ according to the above analysis), we write the indices of the input bits $(\mathcal{S}_\ell \cup \mathcal{R}_\ell)^c=[N]\setminus(\mathcal{S}_\ell \cup \mathcal{R}_\ell)$ as $\{i_1,i_2,...,i_{N-s_\ell-r_\ell}\}$, where $|\mathcal{R}|=r_\ell$ and $|\mathcal{S}_\ell|=s_\ell$, and assume that $i_1<i_2<\cdot\cdot\cdot< i_{N-s_\ell-r_\ell}$. We have
{\allowdisplaybreaks\begin{eqnarray}
\begin{aligned}
I_P\Big(\M_\ell \F_\ell;\Z^{[N]},\X_{1:\ell-1}^{[N]}\Big)&=I_P\Big(\U^{(\mathcal{S}_\ell \cup \mathcal{R}_\ell)^c}_\ell;\Z^{[N]},\X_{1:\ell-1}^{[N]}\Big) \\ \notag
&=I_P\Big(\U_\ell^{i_1}, \U_\ell^{i_2},..., \U_\ell^{i_{N-r_\ell-s_\ell}};\Z^{[N]},\X_{1:\ell-1}^{[N]}\Big) \\
&= \sum_{j=1}^{N-r_\ell-s_\ell} I_P\Big(\U_\ell^{i_j};\Z^{[N]},\X_{1:\ell-1}^{[N]}|\U_\ell^{i_1}, \U_\ell^{i_2},..., \U_\ell^{i_{j-1}}\Big)\\
&= \sum_{j=1}^{N-r_\ell-s_\ell} I_P\Big(\U_\ell^{i_j};\Z^{[N]},\X_{1:\ell-1}^{[N]},\U_\ell^{i_1}, \U_\ell^{i_2},..., \U_\ell^{i_{j-1}}\Big) \\
&\stackrel{(a)}\leq \sum_{j=1}^{N-r_\ell-s_\ell} I_P\Big(\U_\ell^{i_j};\Z^{[N]},\X_{1:\ell-1}^{[N]},\U_\ell^{1}, \U_\ell^{2},..., \U_\ell^{i_{j}-1}\Big), \\
\end{aligned}
\end{eqnarray}}where $(a)$ holds because adding more variables will not decrease the mutual information.

Then the above mutual information can be bounded by the mutual information of the symmetric channel plus an infinitesimal term as follows:
{\allowdisplaybreaks\begin{flalign*}
\sum_{j=1}^{N-r_\ell-s_\ell}&I_P\Big(\U_\ell^{i_j};\Z^{[N]},\X_{1:\ell-1}^{[N]},\U_\ell^{1:i_j-1}\Big)\\
\stackrel{(a)}\leq&\sum_{j=1}^{N-r_\ell-s_\ell} I\Big(\widetilde{\U}_\ell^{i_j};\Z^{[N]},\X_{1:\ell-1}^{[N]},\widetilde{\X}_\ell^{[N]} \oplus \X_\ell^{[N]}, \widetilde{\U}_\ell^{1:i_j-1}\Big)+H\Big(\widetilde{\U}_\ell^{i_j}|\Z^{[N]},\X_{1:\ell-1}^{[N]},\widetilde{\X}_\ell^{[N]} \oplus \X_\ell^{[N]}, \widetilde{\U}_\ell^{1:i_j-1}\Big)\\
-&\sum_{j=1}^{N-r_\ell-s_\ell} H\Big(\U_\ell^{i_j}|\Z^{[N]},\X_{1:\ell-1}^{[N]},\U_\ell^{1:i_j-1}\Big)\\
\stackrel{(b)}\leq& \sum_{j=1}^{N-r_\ell-s_\ell} I\Big(\widetilde{\U}_\ell^{i_j};\Z^{[N]},\X_{1:\ell-1}^{[N]},\widetilde{\X}_\ell^{[N]} \oplus X_\ell^{[N]}, \widetilde{\U}_\ell^{1:i_j-1}\Big)\\
+&\sum_{j=1}^{N-r_\ell-s_\ell} Z\Big(\U_\ell^{i_j}|\Z^{[N]},\X_{1:\ell-1}^{[N]},\U_\ell^{1:i_j-1}\Big)-\Big(Z(\U_\ell^{i_j}|\Z^{[N]},\X_{1:\ell-1}^{[N]},\U_\ell^{1:i_j-1})\Big)^2\\
\stackrel{(c)}\leq& \sum_{j=1}^{N-r_\ell-s_\ell} I\Big(\widetilde{\U}_\ell^{i_j};\Z^{[N]},\X_{1:\ell-1}^{[N]},\widetilde{\X}_\ell^{[N]} \oplus \X_\ell^{[N]}, \widetilde{\U}_\ell^{1:i_j-1}\Big)+N2^{-N^{\beta}} \\
\stackrel{(d)}\leq& \:\:N2^{-N^{\beta'}}+N2^{-N^{\beta}} \\
\leq& \:\:2N2^{-N^{\beta'}}
\end{flalign*}}
for $0<\beta'<\beta<0.5$. Inequalities $(a)$-$(d)$ follow from
\begin{itemize}
\item[] $(a)$ uniformly distributed $\widetilde{\U}_\ell^{i_j}$,
\item[] $(b)$ \cite[Proposition 2]{polarsource} which gives $H(\X|\Y)-H(\X|\Y,\Z)\leq Z(\X|\Y)-(Z(\X|\Y,\Z)^2)$ and Lemma \ref{lem:conasypolar},
\item[] $(c)$ our coding scheme guaranteeing that $Z\Big(\U_\ell^{i_j}|\Z^{[N]},\X_{1:\ell-1}^{[N]},\U_\ell^{1:i_j-1}\Big)$ is greater than $1-2^{-N^{\beta}}$ for the frozen bits and information bits,
\item[] $(d)$ Lemma \ref{lem:Mutualbound}.
\end{itemize}

For wiretap coding, the message $\M_\ell$, frozen bits $\F_\ell$ and random bits $\mathsf{R}_\ell$ are all uniformly random, and the shaping bits $\mathsf{S}_\ell$ are determined by $\mathsf{S}^c_\ell$ according to $\Phi_{\mathcal{S}_\ell}$. Let $Q_{\U_\ell^{[N]}, \X_{1:\ell-1}^{[N]},\Z^{[N]}}$ denote the joint distribution of $(\U_\ell^{[N]}, \X_{1:\ell-1}^{[N]},\Z^{[N]})$ resulted from uniformly distributed $\M_\ell\F_\ell\mathsf{R}_\ell$ and $\mathsf{S}_\ell$ according to $\Phi_{\mathcal{S}_\ell}$. By the proofs of \cite[Th. 5]{polarlatticeJ} and \cite[Th. 6]{polarlatticeJ}, the total variation distance can be bounded as
\begin{eqnarray}\label{eq:vd}
\Big\|Q_{\U_\ell^{[N]}, \X_{1:\ell-1}^{[N]},\Z^{[N]}}-P_{\U_\ell^{[N]}, \X_{1:\ell-1}^{[N]},\Z^{[N]}}\Big\| \leq N 2^{-N^ {\beta'}}
\end{eqnarray}
for sufficiently large $N$.

By \cite[Proposition 5]{howtoachieveAsym}, the mutual information $I(\M_\ell \F_\ell;\Z^{[N]},\X_{1:\ell-1}^{[N]})$ due to $Q_{\U_\ell^{[N]}, \X_{1:\ell-1}^{[N]},\Z^{[N]}}$ satisfies
\begin{eqnarray}
\Big|I(\M_\ell \F_\ell;\Z^{[N]},\X_{1:\ell-1}^{[N]})-I_P(\M_\ell \F_\ell;\Z^{[N]},\X_{1:\ell-1}^{[N]})\Big| &\leq& 7N2^{-N^ {\beta'}}\log 2^N+h_2\Big(N2^{-N^ {\beta'}}\Big)+h_2\Big(4N2^{-N^ {\beta'}}\Big) \notag \\
&=&O\Big(N^2 2^{-N^ {\beta'}}\Big), \notag
\end{eqnarray}
where $h_2(\cdot)$ denotes the binary entropy function.

\end{proof}

Finally, strong secrecy (for uniform message bits) can be proved in the same fashion as shown in \eqref{eqn:upperbound} as:
\begin{eqnarray}\label{eqn:strongsec}
I\Big(\M \F;\Z^{[N]}\Big) \leq \sum_{\ell=1}^{r}I\Big(\M_\ell\F_\ell;\Z^{[N]},\X^{[N]}_{1:\ell-1}\Big)= O\Big(rN^2 2^{-N^{\beta'}}\Big).
\end{eqnarray}
Therefore we conclude that the whole shaping scheme is secure in the sense that the mutual information leakage between $\M$ and $\Z^{[N]}$ vanishes with the block length $N$.

\subsection{Reliability}\label{sec:reliabilityshape}
The reliability analysis in Sect. \ref{sec:reliability} holds for the wiretap coding without shaping. When shaping is involved, the problematic set $\mathcal{D}_\ell$ at each level is included in the shaping set $\mathcal{S}_\ell$ and hence determined by the random mapping $\Phi_{\mathcal{S}_\ell}$. In this subsection, we propose two decoders to achieve reliability for the shaping case. The first one requires a private link between Alice and Bob to share a vanishing fraction of the random mapping $\Phi_{\mathcal{S}_\ell}$ and the second one uses the Markov block coding technique \cite{NewPolarSchemeWiretap} without sharing the random mapping.


\textbf{Decoder 1:} If $\Phi_{\mathcal{S}_\ell}$ is secretly shared between Alice and Bob (we will show in a moment that only a vanishing fraction of $\Phi_{\mathcal{S}_\ell}$ needs to be shared), the bits in $\mathcal{D}_\ell$ can be recovered by Bob simply by the shared mapping but not requiring the Markov block coding technique. By Theorem \ref{theorem:codingtheoremside}, the reliability at each level can be guaranteed by uniformly distributed independent frozen bits and a random mapping $\Phi_{\mathcal{S}_\ell}$ according to $P_{\U_\ell^i|\U_\ell^{1:i-1},X_{1:\ell-1}^{[N]}}$ at each level. The decoding rule is given as follows.
\begin{itemize}
\item Decoding: The decoder receives $y^{[N]}$ and estimates $\widehat{u}_\ell^{[N]}$ based on the previously recovered $x_{1:\ell-1}^{[N]}$ according to the rule
\begin{equation}
\widehat{u}_\ell^i=
\begin{cases}
u_\ell^i, \;\;\;\;\;\;\;\;\;\;\;\;\;\;\;\;\text{if } i\in\mathcal{F}_\ell \\
\phi_i(\widehat{u}_\ell^{1:i-1}, x_{1:\ell-1}^{[N]}),   \;\;\;\;\; \text{if } i\in\mathcal{S}_\ell \\
\underset{u}{\operatorname{argmax}} \; P_{\U_\ell^i|\U_\ell^{1:i-1},\X_{1:\ell-1}^{[N]},\Y^{[N]}}(u|\widehat{u}_\ell^{1:i-1},x_{1:\ell-1}^{[N]},y^{[N]}), \,\text{if } i\in\mathcal{I}_\ell \notag\
\end{cases}.
\label{eqn:decoding}
\end{equation}
\end{itemize}

Note that probability $P_{\U_\ell^i|\U_\ell^{1:i-1},\X_{1:\ell-1}^{[N]},\Y^{[N]}}(u|\widehat{u}_\ell^{1:i-1},x_{1:\ell-1}^{[N]},y^{[N]})$ can be calculated by the SC decoding algorithm efficiently, treating $\Y$ and $\X_{1:\ell-1}$ (already decoded by the SC decoder at previous levels) as the outputs of the asymmetric channel. As a result, the expectation of the decoding error probability over the randomized mappings satisfies $E_{\Phi_{\mathcal{S}_\ell}}[P_e(\phi_{\mathcal{S}_\ell})]=O(2^{-N^{\beta'}})$ for any $\beta'<\beta<0.5$.

Consequently, by the multilevel decoding and union bound, the expectation of the block error probability of our wiretap coding scheme is vanishing as $N \rightarrow \infty$. However, this result is based on the assumption that the mapping $\Phi_{\mathcal{S}_\ell}$ is only shared between Alice and Bob. To share this mapping, we can let Alice and Bob have access to the same source of randomness, which may be achieved by a private link between Alice and Bob. \color{black}Fortunately, the rate of this private link can be made vanishing since the proportion of the shaping bits covered by the mapping $\Phi_{\mathcal{S}_\ell}$ can be significantly reduced.

Recall that the shaping set $\mathcal{S}_\ell$ is defined by
\begin{eqnarray}
\mathcal{S}_\ell \triangleq \Big\{i \in [N]: Z(\U_{\ell}^i|\U_{\ell}^{1:i-1},\X_{1:\ell-1}^{[N]}) < 1-2^{-N^\beta} \,\text{or}\,\, 2^{-N^{\beta}}<Z(\U_\ell^i|\U_\ell^{1:i-1},\Y^{[N]},\X_{1:\ell-1}^{[N]})<1-2^{-N^{\beta}}\Big\}.
\end{eqnarray}
It has been shown in \cite[Th. 2]{polarlatticeQZ} and \cite[Th. 15]{polarflashmemo} that the shaping bits in the subset  $\{i \in [N]: Z(\U_{\ell}^i|\U_{\ell}^{1:i-1},\X_{1:\ell-1}^{[N]}) \leq 2^{-N^\beta} \}$ can be recovered according to the rule
\begin{eqnarray}\label{eqn:MAPshaping}
u_\ell^i= \argmax_u P_{\U_\ell^i|\U_\ell^{1:i-1},\X_{1:\ell-1}^{1:N}}(u|u_\ell^{1:i-1},x_{1:\ell-1}^{1:N}) \,\,\,\,\,\text{if} \,\,\,\,\, Z(\U_{\ell}^i|\U_{\ell}^{1:i-1},\X_{1:\ell-1}^{[N]}) \leq 2^{-N^\beta}, \notag
\end{eqnarray}
instead of mapping.  This modification has negligible impact on strong secrecy. Let us explain it briefly. For the shaping bits in $\mathcal{S}_\ell$ with $Z(\U_{\ell}^i|\U_{\ell}^{1:i-1},\X_{1:\ell-1}^{[N]}) \leq 2^{-N^\beta}$, we also have $H(\U_{\ell}^i|\U_{\ell}^{1:i-1},\X_{1:\ell-1}^{[N]}) \leq 2^{-N^\beta}$. This means that $\U_{\ell}^i$ in $\mathcal{S}_\ell$ is almost determined by $\U_{\ell}^{1:i-1}$ and $\X_{1:\ell-1}^{[N]}$ when $N$ is sufficiently large. The probability $P_{\U_\ell^i|\U_\ell^{1:i-1},\X_{1:\ell-1}^{1:N}}(u|u_\ell^{1:i-1},x_{1:\ell-1}^{1:N})$ for those bits can be arbitrarily close to either 0 or 1. Therefore, replacing the random rounding rule with the MAP decision rule for those bits will yield another vanishing term $N2^{-N^{\beta'}}$ on the right hand side of the upper bound of the total variation distance as shown in \eqref{eq:vd}, which results in negligible difference on the information leakage when $N$ grows large. Moreover, since $Z(\U_{\ell}^i|\U_{\ell}^{1:i-1},\X_{1:\ell-1}^{[N]}) \leq 2^{-N^\beta}$ for theses shaping bits, using the MAP decision rule will also yield an additional vanishing term $N2^{-N^{\beta'}}$ on the upper bound of the decoding error probability for Bob.
\color{black}
As a result, the deterministic mapping has only to cover the unpolarized set
\begin{eqnarray}\label{eqn:setdS}
d\mathcal{S}_\ell=\Big\{i\in[N]:2^{-N^{\beta}}<Z(\U_\ell^i|\U_\ell^{1:i-1},\X_{1:\ell-1}^{1:N})<1-2^{-N^{\beta}}\text{ or } \notag \\
2^{-N^{\beta}}<Z(\U_\ell^i|\U_\ell^{1:i-1},\Y^{1:N},\X_{1:\ell-1}^{1:N})<1-2^{-N^{\beta}}\Big\}, \notag
\end{eqnarray}
whose proportion $\frac{|d\mathcal{S}_\ell|}{N} \to 0$ as $N\to \infty$.

\begin{rem}
By the channel equivalence, when $\Phi_{\mathcal{S}_\ell}$ is shared to Bob, the decoding of $\Lambda_b$ is equivalent to the MMSE lattice decoding proposed in \cite{cong2} for random lattice codes. When instantiated with a polar lattice, we use multistage lattice decoding. More explicitly, by \cite[Lemma 7]{polarlatticeJ}, the SC decoding of the asymmetric channel can be converted to the SC decoding of its symmetrized channel, which is equivalent to the MMSE-scaled partition channel in the lattice Gaussian shaping case \cite[Lemma 9]{polarlatticeJ}.
\end{rem}

\textbf{Decoder 2:} Alternatively, one can also use the block Markov coding technique \cite{NewPolarSchemeWiretap} to achieve reliability without sharing $\Phi_{\mathcal{S}_\ell}$. As shown in Fig. \ref{fig:markovblock}, the message at $\ell$-th level is divided into $k_\ell$ blocks. Denote by $\Delta\mathsf{S}_\ell$ the bits in unpolarized set $d\mathcal{S}_\ell$. The shaping bits $\mathsf{S}_\ell$ for each block is further divided into unpolarized bits $\Delta\mathsf{S}_\ell$ and polarized shaping bits $\mathsf{S}_\ell \setminus \Delta\mathsf{S}_\ell$. As mentioned above, only $\Delta\mathsf{S}_\ell$ needs to be covered by mapping and its proportion is vanishing. We can sacrifice some message bits to convey $\Delta\mathsf{S}_\ell$ for the next block without involving significant rate loss. These wasted message bits are denoted by $\mathsf{E}_\ell$. For encoding, we start with the last block (Block $k_\ell$). Given $\F_\ell$, $\M_\ell$ (no $\mathsf{E}_\ell$ for the last block) and $\mathsf{R}_\ell$, we can obtain $\Delta\mathsf{S}_\ell$ according to $\Phi_{\mathcal{S}_\ell}$. Then we copy $\Delta\mathsf{S}_\ell$ of the last block to the bits $\mathsf{E}_\ell$ of its previous block and do encoding to get the $\Delta\mathsf{S}_\ell$ of block $k_\ell-1$. This process ends until we get the $\Delta\mathsf{S}_\ell$ of the first block. This scheme is similar to the one we discussed in Sect. \ref{sec:reliability}. To achieve reliability, we need a secure code with vanishing rate to convey the bits $\Delta\mathsf{S}_\ell$ of the first block to Bob. See \cite{MatthieuBlochCovert} for an example of such codes. To guarantee an insignificant rate loss, $k_\ell$ is required to be sufficiently large. We may set $k_\ell=O(N^\alpha)$ for some $\alpha>0$.

\begin{figure}[h]
    \centering
    \includegraphics[width=10cm]{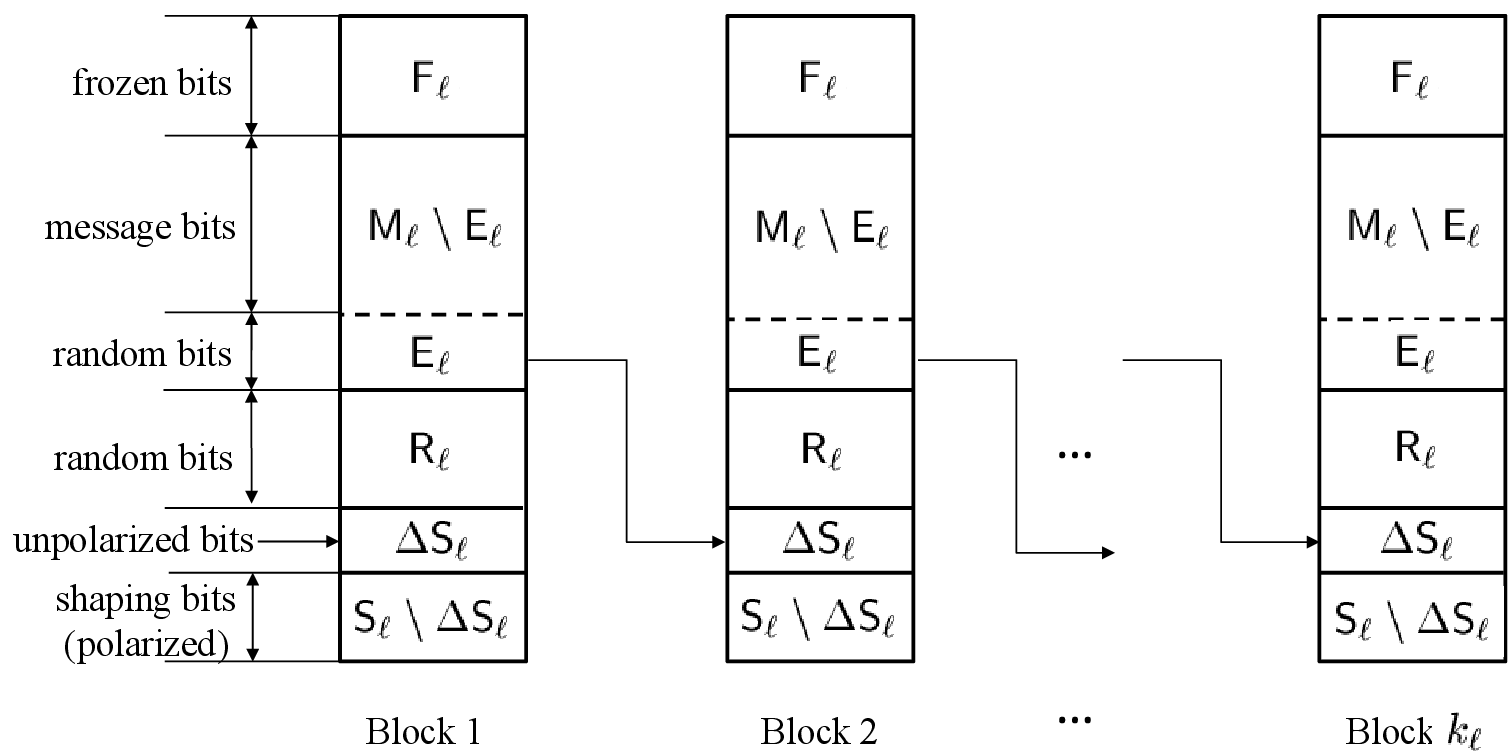}
    \caption{Markov block coding scheme without sharing the secret mapping.}
    \label{fig:markovblock}
\end{figure}

Now we present the main theorem of the paper.

\begin{ther}[Achieving secrecy capacity of the GWC]
Consider a multilevel lattice code constructed from polar codes based on asymmetric channels and lattice Gaussian shaping $D_{\Lambda,\sigma_s}$. Given $\sigma_e^2>\sigma_b^2$, let $\epsilon_{\Lambda}(\widetilde{\sigma}_e)$ be negligible and set the number of levels $r=O(\log\log N)$ for $N\rightarrow\infty$. Then all strong secrecy rates $R$ satisfying $R<\frac{1}{2}\log\left(\frac{1+\SNR_b}{1+\SNR_e}\right)$ are achievable for the Gaussian wiretap channel, where $\SNR_b$ and $\SNR_e$ denote the $\SNR$ of the main channel and wiretapper's channel, respectively.
\end{ther}
\begin{proof}
The reliability condition and the strong secrecy condition are satisfied by Theorem \ref{theorem:codingtheoremside} and Lemma \ref{lem:secrecybound}, respectively. It remains to illustrate that the secrecy rate approaches the secrecy capacity. For some $\epsilon'\to 0$, we have
\begin{eqnarray}
\begin{aligned}
\lim_{N\rightarrow\infty}R&=\sum_{\ell=1}^{r}\lim_{N\rightarrow\infty}\frac{|\mathcal{A}_\ell^{\mathcal{S}^c}|}{N} \\
&=\sum_{\ell=1}^{r}I(\X_\ell;\Y|\X_1,\cdot\cdot\cdot,\X_{\ell-1})-I(\X_\ell;\Z|\X_1,\cdot\cdot\cdot,\X_{\ell-1}) \\
&\stackrel{(a)}=\frac{1}{2}\log\left(\frac{\widetilde{\sigma}_e^2}{\widetilde{\sigma}_b^2}\right)-\epsilon'\\
&\stackrel{(b)}\geq\frac{1}{2}\log\left(\frac{1+\SNR_b}{1+\SNR_e}\right)-\epsilon',
\end{aligned}
\end{eqnarray}
where ($a$) is due to Lemma \ref{lem:shapdoesnotmatter}, and ($b$) is because the signal power $P_s\leq \sigma_s^2$ \cite[Lemma 1]{LingBel13}\footnote{Of course, $R$ cannot exceed the secrecy capacity, so this inequality implies that $P_s \to \sigma_s^2$.}, respectively.
\end{proof}

\subsection{Semantic security}\label{sec:semantic}

In this subsection, we extend strong secrecy of the constructed polar lattices to semantic security, namely the resulted strong secrecy does not rely on the distribution of the message. We take the level-1 wiretapper's channel $W_1$ as an example. Our goal is to show that the maximum mutual information between $\M_1 \F_1$ and $\Z^{[N]}$ is vanishing for any input distribution as $N \rightarrow \infty$. Unlike the symmetric randomness induced channel introduced in \cite{polarsecrecy}, the new induced channel is generally asymmetric with transition probability
\begin{equation}
Q(z|v)=\frac{1}{2^{r_1}} \sum_{\Phi_{\mathcal{S}_1}}P(\Phi_{\mathcal{S}_1})\sum_{e \in \{0,1\}^{r_1}} W_1^N(z|(v,e,\Phi_{\mathcal{S}_1}(v,e))G_N), \notag
\end{equation}
where $\Phi_{\mathcal{S}_1}(v,e)$ represents the shaping bits determined by $v$ (the frozen bits and message bits together) and $e$ (the random bits) according to the random mapping $\Phi_{\mathcal{S}_1}$. It is difficult to find the optimal input distribution to maximize the mutual information for the new induced channel.

To prove the semantic security, we investigate the relationship between the $i$-th subchannel of $W_{1,N}$ and the $i$-th subchannel of its symmetrized version $\widetilde{W}_{1,N}$, which are denoted by $W_{1}^{(i,N)}$ and $\widetilde{W}_{1}^{(i,N)}$, respectively. According to Lemma \ref{lem:Asy2sym}, the asymmetric wiretap channel $W_1: \X_1 \rightarrow \Z$ is symmetrized to channel $\widetilde{W}_1: \widetilde{\X}_1 \rightarrow (Z, \widetilde{\X}_1 \oplus \X_1)$. After the $N$-by-$N$ polarization transform, we obtain $W_{1}^{(i,N)}: \U_1^i \rightarrow (\U_1^{1:i-1}, \Z^{[N]})$ and $\widetilde{W}_{1}^{(i,N)}: \widetilde{\U}_1^i \rightarrow (\widetilde{\U}_1^{1:i-1}, \widetilde{\X}_1^{[N]} \oplus \X_1^{[N]}, \Z^{[N]} )$. The next lemma shows that if we symmetrize $W_{1}^{(i,N)}$ directly, i.e., construct a symmetric channel $\widetilde{W_{1}^{(i,N)}}: \widetilde{\U}_1^i \rightarrow (\U_1^{1:i-1}, \Z^{[N]}, \widetilde{\U}_1^i \oplus \U_1^i)$ in the sense of Lemma \ref{lem:Asy2sym}, $\widetilde{W_{1}^{(i,N)}}$ is degraded with respect to $\widetilde{W}_{1}^{(i,N)}$.

\begin{lem}\label{lem:semanticDegrad}
The symmetrized channel $\widetilde{W_{1}^{(i,N)}}$ derived directly from $W_{1}^{(i,N)}$ is degraded with respect to the $i$-th subchannel $\widetilde{W}_{1}^{(i,N)}$ of $\widetilde{W}_1$.
\end{lem}
\begin{proof}
According to the proof of \cite[Theorem 2]{aspolarcodes}, we have the relationship
\begin{equation}
\widetilde{W}_1^{(i,N)}(\widetilde{u}_1^{1:i-1},  \widetilde{x}_1^{[N]}\oplus x_1^{[N]}, z^{[N]}|\widetilde{u}_1^i)=2^{-N+1}P_{\U_1^{1:i},\Z^{[N]}}(u_1^{1:i},z^{[N]}). \notag
\end{equation}
Letting $\widetilde{x}_1^{[N]}\oplus x_1^{[N]}=0^{[N]}$, the equation becomes $\widetilde{W}_1^{(i,N)}(u_1^{1:i-1}, 0^{[N]}, z^{[N]}     |u_1^i)=2^{-N+1}P_{\U_1^{1:i},\Z^{[N]}}(u_1^{1:i},z^{[N]})$, which has already been addressed in \cite{aspolarcodes}. However, for a fixed $x_1^{[N]}$ and $\widetilde{u}_1^i=u_1^i$, since $G_N$ is full rank, there are $2^{N-1}$ choices of $\widetilde{x}_1^{[N]}$ remaining, which means that there exists $2^{N-1}$ outputs symbols of $\widetilde{W}_1^{(i,N)}$ having the same transition probability $2^{-N+1}P_{\U_1^{1:i},\Z^{[N]}}(u_1^{1:i},z^{[N]})$. Suppose a middle channel which maps all these output symbols to one single symbol, which is with transition probability $P_{\U_1^{1:i},\Z^{[N]}}(u_1^{1:i},z^{[N]})$. The same operation can be done for $\widetilde{u}_1^i=u_1^i\oplus 1$, making another symbol with transition probability $P_{\U_1^{1:i},\Z^{[N]}}(u_1^{1:i},z^{[N]})$ corresponding to the input $u_1^i\oplus 1$. This is a channel degradation process, and the degraded channel is symmetric.

Then we show that the symmetrized channel $\widetilde{W_{1}^{(i,N)}}$ is equivalent to the degraded channel mentioned above.
By Lemma \ref{lem:Asy2sym}, the channel transition probability of $\widetilde{W_{1}^{(i,N)}}$ is
\begin{equation}
\widetilde{W_{1}^{(i,N)}}(u_1^{1:i-1},  \widetilde{u}_1^{i}\oplus u_1^{i}, z^{[N]}|\widetilde{u}_1^i)=P_{\U_1^{1:i},\Z^{[N]}}(u_1^{1:i},z^{[N]}), \notag
\end{equation}
which is equal to the transition probability of the degraded channel discussed in the previous paragraph. Therefore, $\widetilde{W_{1}^{(i,N)}}$ is degraded with respect to $\widetilde{W}_{1}^{(i,N)}$.
\end{proof}

\begin{rem}
In fact, a stronger relationship that $\widetilde{W_{1}^{(i,N)}}$ is equivalent to $\widetilde{W}_{1}^{(i,N)}$ can be proved. This is because that the output symbols combined in the channel degradation process have the same LR. An evidence of this result can be found in \cite[Equation (36)]{aspolarcodes}, where $\widetilde{Z}(\widetilde{W}_{1}^{(i,N)})=Z(\U_1^i|\U_1^{1:i-1},\Z^{[N]})=\widetilde{Z}(\widetilde{W_{1}^{(i,N)}})$. Nevertheless, the degradation relationship is sufficient for this work. Notice that Lemma \ref{lem:semanticDegrad} can be generalized to high level $\ell$, with outputs $\Z^{[N]}$ replaced by $(\Z^{[N]}, \X_{1:\ell-1}^{[N]})$.
\end{rem}

Illuminated by Lemma \ref{lem:semanticDegrad}, we can also symmetrize the new induced channel at level $\ell$ and show that it is degraded with respect to the randomness-induced channel constructed from $\widetilde{W}_{\ell}$. For simplicity, letting $\ell=1$, the new induced channel at level 1 is $\mathcal{Q}_N(W_1,\mathcal{S}_1,\mathcal{R}_1): \U_1^{(\mathcal{S}_1\cup\mathcal{R}_1)^c}\rightarrow \Z^{[N]}$, which is symmetrized to $\widetilde{\mathcal{Q}}_N(W_1,\mathcal{S}_1,\mathcal{R}_1): \widetilde{\U}_1^{(\mathcal{S}_1\cup\mathcal{R}_1)^c}\rightarrow (\Z^{[N]}, \widetilde{\U}_1^{(\mathcal{S}_1\cup\mathcal{R}_1)^c} \oplus \U_1^{(\mathcal{S}_1\cup\mathcal{R}_1)^c})$ in the same fashion as in Lemma \ref{lem:Asy2sym}. Recall that the randomness-induced channel of $\widetilde{W}_1$ defined in \cite{polarsecrecy} can be denoted as $\mathcal{Q}_N(\widetilde{W}_1,\mathcal{R}_1 \cup \mathcal{S}_1): \widetilde{\U}_1^{(\mathcal{S}_1\cup\mathcal{R}_1)^c}\rightarrow (\Z^{[N]}, \widetilde{\X}_1^{[N]} \oplus \X_1^{[N]})$. Note that for the randomness-induced channel $\mathcal{Q}_N(\widetilde{W}_1,\mathcal{R}_1 \cup \mathcal{S}_1)$, set $\mathcal{R}_1 \cup \mathcal{S}_1$ is fed with uniformly random bits, which is different from the shaping-induced channel.
\begin{lem}\label{lem:semanticDegrad1}
For an asymmetric channel $W_1: \X_1 \rightarrow \Z$ and its symmetrized channel $\widetilde{W}_1: \widetilde{\X}_1 \rightarrow (\Z, \widetilde{\X}_1\oplus\X_1)$, the symmetrized version of the new induced channel $\widetilde{\mathcal{Q}}_N(W_1,\mathcal{S}_1,\mathcal{R}_1)$ is degraded with respect to the randomness-induced channel $\mathcal{Q}_N(\widetilde{W}_1,\mathcal{R}_1 \cup \mathcal{S}_1)$.
\end{lem}
\begin{proof}
The proof is similar to that of Lemma \ref{lem:semanticDegrad}. For a fixed realization $x_1^{[N]}$ and input $\widetilde{u}_1^{(\mathcal{S}_1\cup\mathcal{R}_1)^c}$, there are $2^{|\mathcal{S}_1\cup\mathcal{R}_1|}$ choice of $\widetilde{x}_1^{[N]}$ remaining. Since $z^{[N]}$ is only dependent on $x_1^{[N]}$, we can build a middle channel which merges the $2^{|\mathcal{S}_1\cup\mathcal{R}_1|}$ output symbols of $\mathcal{Q}_N(\widetilde{W}_1,\mathcal{R}_1 \cup \mathcal{S}_1)$ to one output symbol of $\widetilde{\mathcal{Q}}_N(W_1,\mathcal{S}_1,\mathcal{R}_1)$, which means that $\widetilde{\mathcal{Q}}_N(W_1,\mathcal{S}_1,\mathcal{R}_1)$ is degraded with respect to $\mathcal{Q}_N(\widetilde{W}_1,\mathcal{R}_1 \cup \mathcal{S}_1)$. Again, this result can be generalized to higher levels.
\end{proof}

Finally, we are ready to prove the semantic security of our wiretap coding scheme. For brevity, let $\M_\ell\F_\ell$ and $\widetilde{\M}_\ell\widetilde{\F}_\ell$ denote $\U^{(\mathcal{S}_\ell \cup \mathcal{R}_\ell)^c}_\ell$ and $\widetilde{\U}^{(\mathcal{S}_\ell \cup \mathcal{R}_\ell)^c}_\ell$, respectively. Recall that $\M$ is divided into $\M_1, ..., \M_r$ at each level. We express $\M\F$ and $\widetilde{\M}\widetilde{\F}$ as the collection of message and frozen bits on all levels of the new induced channel and the symmetric randomness-induced channel, respectively. We also define $\widetilde{\M}\widetilde{\F}\oplus \M\F$ as the operation $\widetilde{\M}_\ell \widetilde{\F}_\ell \oplus \M_\ell\F_\ell$ from level $1$ to level $r$.

\begin{ther}[Semantic security]\label{lem:semanticBound}
For arbitrarily distributed message $\M$, the information leakage $I(\M; \Z^{[N]})$ of the proposed wiretap lattice code is upper-bounded as
\begin{eqnarray}\notag
I\Big(\M; \Z^{[N]}\Big) \leq I\Big(\widetilde{\M}\widetilde{\F}; \Z^{[N]}, \widetilde{\M}\widetilde{\F}\oplus \M\F\Big) \leq rN2^{-N^{\beta'}},
\end{eqnarray} where $I\Big(\widetilde{\M}\widetilde{\F}; \Z^{[N]}, \widetilde{\M}\widetilde{\F}\oplus \M\F\Big)$ is the capacity of the symmetrized channel derived from the non-binary channel $\M\F \rightarrow \Z^{[N]}$ \footnote{The symmetrization of a non-binary channel is similar to that of a binary channel as shown in Lemma \ref{lem:Asy2sym}. When $\X$ and $\widetilde{\X}$ are both non-binary, $\X \oplus \widetilde{\X}$ denotes the result of the exclusive or (xor) operation of the binary expressions of $\X$ and $\widetilde{\X}$.}.
\end{ther}
\begin{proof}
\sloppy
By \cite[Proposition 16]{polarsecrecy}, the channel capacity of the randomness-induced channel $\mathcal{Q}_N(\widetilde{W}_1,\mathcal{S}_1,\mathcal{R}_1)$ is upper-bounded by $N2^{-N^{\beta'}}$ when partition rule \eqref{eqn:Good&bad} is used. By channel degradation, the channel capacity of the symmetrized new induced channel $\widetilde{\mathcal{Q}}_N(W_1,\mathcal{S}_1,\mathcal{R}_1)$ can also be upper-bounded by $N2^{-N^{\beta'}}$. Since this result can be generalized to higher level $\ell$ ($\ell \geq 1 $), we obtain $C(\widetilde{\mathcal{Q}}_N(W_\ell,\mathcal{S}_\ell,\mathcal{R}_\ell)) \leq N2^{-N^{\beta'}}$, which means
$I\Big(\widetilde{\M}_\ell\widetilde{\F}_\ell;\Z^{[N]},\X^{[N]}_{1:\ell-1}, \widetilde{\M}_{\ell}\widetilde{\F}_{\ell} \oplus \M_{\ell}\F_{\ell}\Big) \leq N2^{-N^{\beta'}}$. Similarly to \eqref{eqn:upperbound}, we have
{\allowdisplaybreaks\begin{flalign*}\label{eqn:upperboundend}
&I\Big(\widetilde{\M}\widetilde{\F};\Z^{[N]}, \widetilde{\M}\widetilde{\F}\oplus \M\F\Big)  \\
&=\sum_{\ell=1}^{r}I\Big(\widetilde{\M}_\ell\widetilde{\F}_\ell;\Z^{[N]}, \widetilde{\M}\widetilde{\F}\oplus \M\F|\widetilde{\M}_{1:\ell-1}\widetilde{\F}_{1:\ell-1}\Big) \\
&=\sum_{\ell=1}^{r}H\Big(\widetilde{\M}_\ell\widetilde{\F}_\ell|\widetilde{\M}_{1:\ell-1}\widetilde{\F}_{1:\ell-1}\Big)-H\Big(\widetilde{\M}_\ell\widetilde{\F}_\ell|\Z^{[N]}, \widetilde{\M}\widetilde{\F}\oplus \M\F ,\widetilde{\M}_{1:\ell-1}\widetilde{\F}_{1:\ell-1}\Big)\\
&\leq\sum_{\ell=1}^{r}H\Big(\widetilde{\M}_\ell\widetilde{\F}_\ell\Big)-H\Big(\widetilde{\M}_\ell\widetilde{\F}_\ell|\Z^{[N]}, \widetilde{\M}\widetilde{\F}\oplus \M\F ,\widetilde{\M}_{1:\ell-1}\widetilde{\F}_{1:\ell-1}\Big)\\
&=\sum_{\ell=1}^{r}I\Big(\widetilde{\M}_\ell\widetilde{\F}_\ell;\Z^{[N]},\widetilde{\M}\widetilde{\F}\oplus \M\F ,\widetilde{\M}_{1:\ell-1}\widetilde{\F}_{1:\ell-1}\Big)\\
&\stackrel{(a)}=\sum_{\ell=1}^{r}I\Big(\widetilde{\M}_\ell\widetilde{\F}_\ell;\Z^{[N]},\M_{1:\ell-1}\F_{1:\ell-1} ,\widetilde{\M}_{\ell}\widetilde{\F}_{\ell} \oplus \M_{\ell}\F_{\ell}\Big)\\
&\stackrel{(b)}\leq \sum_{\ell=1}^{r}I\Big(\widetilde{\M}_\ell\widetilde{\F}_\ell;\Z^{[N]},\X^{[N]}_{1:\ell-1}, \widetilde{\M}_{\ell}\widetilde{\F}_{\ell} \oplus \M_{\ell}\F_{\ell}\Big) \\
&\leq rN2^{-N^{\beta'}},
\end{flalign*}}where equality $(a)$ holds because $\Z^{[N]}$ is determined by $\M\F\mathsf{R}$ and $\widetilde{\M}_\ell\widetilde{\F}_\ell$ is independent of $\widetilde{\M}_{\ell+1:r}\widetilde{\F}_{\ell+1:r} \oplus \M_{\ell+1:r}\F_{\ell+1:r}$, and inequality $(b)$ holds because adding more variables will not decrease the mutual information.

Therefore, we have
\begin{eqnarray}\begin{aligned} \notag\
I\Big(\M; \Z^{[N]}\Big) & \leq I\Big(\M\F; \Z^{[N]}\Big) \\
& \stackrel{(a)} \leq H\Big(\widetilde{\M}\widetilde{\F}\oplus \M\F\Big)-H(\M\F)+ I\Big(\M\F; \Z^{[N]}\Big) \\
& \stackrel{(b)} = I\Big(\widetilde{\M}\widetilde{\F};\Z^{[N]}, \widetilde{\M}\widetilde{\F}\oplus \M\F\Big) \\
& \leq  rN2^{-N^{\beta'}},
\end{aligned}\end{eqnarray}
where the equality in $(a)$ holds iff $\M\F$ is also uniform, and $(b)$ is due to the chain rule.
\end{proof}


\section{Discussion}
We would like to elucidate our coding scheme for the Gaussian wiretap channel in terms of the lattice structure. In Sect. \ref{sec:NoPower}, we constructed the AWGN-good lattice $\Lambda_b$ and the secrecy-good lattice $\Lambda_e$ without considering the power constraint. When the power constraint is taken into consideration, the lattice Gaussian shaping was implemented in Sect. \ref{sec:SecrecyGoodShap}. $\Lambda_b$ and $\Lambda_e$ were then constructed according to the MMSE-scaled main channel and wiretapper's channel, respectively. We note that these two lattices themselves are generated only if the independent frozen bits on all levels are $0$s. Since the independent frozen set of the polar codes at each level is filled with random bits, we actually obtain a coset $\Lambda_b+\chi$ of $\Lambda_b$ and a coset $\Lambda_e+\chi$ of $\Lambda_e$ simultaneously, where $\chi$ is a uniformly distributed shift. This is because we are unable to fix the independent frozen bits $\F_{\ell}$ in our scheme (due to the lack of the proof that the shaping-induced channel is symmetric). By using the lattice Gaussian $D_{\Lambda,\sigma_s}$ as our constellation in each lattice dimension, we would obtain $D_{\Lambda^N,\sigma_s}$ without coding. Since $\Lambda_e+\chi \subset \Lambda_b+\chi \subset \Lambda^N$, we actually implemented the lattice Gaussian shaping over both $\Lambda_b+\chi$ and $\Lambda_e+\chi$.  To summarize, Alice firstly assigns each message $m\in \mathcal{M}$ to a coset $\widetilde{\lambda}_m \in \Lambda_b/\Lambda_e$, then randomly sends a point in the coset $\Lambda_e+\chi+\lambda_m$ ($\lambda_m$ is the coset leader of $\widetilde{\lambda}_m$) according to the distribution $D_{\Lambda_e+\chi+\lambda_m, \sigma_s}$. This scheme is consistent with the theoretical model proposed in \cite{cong2}.

On the mod-$\Lambda_s$ wiretap channel, semantic security was obtained for free due to the channel symmetry. On the power-constrained wiretap channel, a symmetrized new induced channel from $\widetilde{\M}\widetilde{\F}$ to $(\Z^{[N]}, \widetilde{\M}\widetilde{\F}\oplus \M\F)$ was constructed to upper-bound the information leakage. This channel is directly derived from the new induced channel from $\M\F$ to $\Z^{[N]}$. According to Lemma \ref{lem:semanticDegrad}, this symmetrized new induced channel is degraded with respect to the symmetric randomness-induced channel from $\widetilde{\M}\widetilde{\F}$ to $(\Z^{[N]}, \widetilde{\X}^{[N]}_{1:r} \oplus \X^{[N]}_{1:r})$. Moreover, when $\widetilde{\F}$ is frozen, the randomness-induced channel from $\widetilde{\M}$ to $(\Z^{[N]}, \widetilde{\X}^{[N]}_{1:r} \oplus \X^{[N]}_{1:r})$ corresponds to the $\Lambda_b/\Lambda_e$ channel given in Sect. \ref{sec:NoPower} (with MMSE scaling).
\color{black}

\appendices

\section{Proof of Lemma \ref{lem:strongbound}} \label{appendix0}
\begin{proof}
It is sufficient to show $I(\M\F;\Z^{[N]}) \leq N\cdot 2^{-N^{\beta'}}$ since $I(\M;\Z^{[N]}) \leq I(\M\F;\Z^{[N]})$. As has been shown in \cite{polarsecrecy}, the induced channel $\M\F \rightarrow \Z^{[N]}$ is symmetric when $\mathcal{B}$ and $\mathcal{D}$ are fed with random bits $\mathsf{R}$. For a symmetric channel, the maximum mutual information is achieved by uniform input distribution. Let $\widetilde{\U}^\mathcal{A}$ and $\widetilde{\U}^\mathcal{C}$ denote independent and uniform versions of $\M$ and $\F$ and $\widetilde{\Z}^{[N]}$ be the corresponding channel output. Assuming $i_1<i_2<...<i_{|\mathcal{A}\cup \mathcal{C}|}$ are the indices in $\mathcal{A} \cup \mathcal{C}$,
{\allowdisplaybreaks\begin{eqnarray}
I(\M\F;\Z^{[N]}) &\leq& I(\widetilde{\U}^\mathcal{A}\widetilde{\U}^\mathcal{C}; \widetilde{\Z}^{[N]}) \notag \\
 &=& \sum_{j=1}^{|\mathcal{A}\cup \mathcal{C}|} I(\widetilde{\U}^{i_j};\widetilde{\Z}^{[N]}|\widetilde{\U}^{i_1},...,\widetilde{\U}^{i_{j-1}})\notag \\
 &=& \sum_{j=1}^{|\mathcal{A}\cup \mathcal{C}|} I(\widetilde{\U}^{i_j};\widetilde{\Z}^{[N]},\widetilde{\U}^{i_1},...,\widetilde{\U}^{i_{j-1}}) \notag \\
 &\leq& \sum_{j=1}^{|\mathcal{A}\cup \mathcal{C}|} I(\widetilde{\U}^{i_j};\widetilde{\Z}^{[N]},\widetilde{\U}^{1:i_j-1}) \notag \\
 &=& \sum_{j=1}^{|\mathcal{A}\cup \mathcal{C}|} I(\widetilde{W}_N^{(i_j)}) \leq N\cdot 2^{-N^{\beta'}}. \notag
\end{eqnarray}}
\end{proof}

\section{Proof of Lemma \ref{lem:securerate}} \label{appendix1}
\begin{proof}
According to the definitions of $\mathcal{G}(\widetilde{V})$ and $\mathcal{N}(\widetilde{W})$ presented in \eqref{eqn:Good&bad},
\begin{eqnarray}
\lim_{N\rightarrow \infty} \frac{|\mathcal{G}(\widetilde{V})|}{N}&=& \lim_{N\rightarrow \infty} \frac{1}{N} |\{i: \widetilde{Z}(\widetilde{V}_N^{(i)})\leq 2^{-N^\beta}\}|=C(\widetilde{V}), \notag \\
\lim_{N\rightarrow \infty} \frac{|\mathcal{N}(\widetilde{W})|}{N}&=& \lim_{N\rightarrow \infty} \frac{1}{N} |\{i: \widetilde{Z}(\widetilde{W}_N^{(i)})\geq 1-2^{-N^\beta}\}|=1-C(\widetilde{W}). \notag
\end{eqnarray}
Here we define another two sets $\bar{\mathcal{G}}(\widetilde{V})$ and $\bar{\mathcal{N}}(\widetilde{W})$ as
\begin{eqnarray}
\bar{\mathcal{G}}(\widetilde{V})&=\{i:\widetilde{Z}(\widetilde{V}_{N}^{(i)}) \geq 1-2^{-N^\beta}\}, \notag\ \\
\bar{\mathcal{N}}(\widetilde{W})&=\{i:\widetilde{Z}(\widetilde{W}_{N}^{(i)}) \leq 2^{-N^\beta}\}. \notag
\end{eqnarray}
Similarly, we have $\lim_{N\rightarrow \infty} \frac{|\bar{\mathcal{G}}(\widetilde{V})|}{N}=1-C(\widetilde{V})$ and $\lim_{N\rightarrow \infty} \frac{|\bar{\mathcal{N}}(\widetilde{W})|}{N}=C(\widetilde{W})$. Since $\widetilde{W}$ is stochastically degraded with respect to $\widetilde{V}$, $\bar{\mathcal{G}}(\widetilde{V})$ and $\bar{\mathcal{N}}(\widetilde{W})$ are disjoint with each other \cite{polarchannelandsource}, then we have
\begin{eqnarray}
\lim_{N\rightarrow \infty} \frac{|\bar{\mathcal{G}}(\widetilde{V})\cup \bar{\mathcal{N}}(\widetilde{W})|}{N}=1-C(\widetilde{V})+C(\widetilde{W}). \notag
\end{eqnarray}
By the property of polarization, the proportion of the unpolarized part is vanishing as $N$ goes to infinity, i.e.,
\begin{eqnarray}
\lim_{N\rightarrow \infty} \frac{|\mathcal{G}(\widetilde{V})\cup \bar{\mathcal{G}}(\widetilde{V})|}{N}=1, \notag \\
\lim_{N\rightarrow \infty} \frac{|\mathcal{N}(\widetilde{W})\cup \bar{\mathcal{N}}(\widetilde{W})|}{N}=1, \notag
\end{eqnarray}
Finally, we have
\begin{eqnarray}
\lim_{N\rightarrow \infty} \frac{|\mathcal{G}(\widetilde{V})\cap \mathcal{N}(\widetilde{W})|}{N}=1-\lim_{N\rightarrow \infty} \frac{|\bar{\mathcal{G}}(\widetilde{V})\cup \bar{\mathcal{N}}(\widetilde{W})|}{N}=C(\widetilde{V})-C(\widetilde{W}). \notag
\end{eqnarray}
\end{proof}

\section{Proof of Lemma \ref{lem:channelequ}} \label{appendix2}
\begin{proof}
It is sufficient to demonstrate that channel $W(\Lambda_{\ell-1}/\Lambda_{\ell},\sigma_e^2)$ is degraded with respect to $W'(\X_{\ell};\Z|\X_{1:\ell-1})$ and $W'(\X_{\ell};\Z|\X_{1:\ell-1})$ is degraded with respect to $W(\Lambda_{\ell-1}/\Lambda_{\ell},\sigma_e^2)$ as well. To see this, we firstly construct a middle channel $\widehat{W}$ from $\Z \in \mathcal{V}(\Lambda_r)$ to $\bar{\Z} \in \mathcal{V}(\Lambda_\ell)$. For a specific realization $\bar{z}$ of $\bar{\Z}$, this $\widehat{W}$ maps $\bar{z}+[\Lambda_\ell/\Lambda_r]$ to $\bar{z}$ with probability 1, where $[\Lambda_\ell/\Lambda_r]$ represents the set of the coset leaders of the partition $\Lambda_\ell/\Lambda_r$. Then we obtain channel $W(\Lambda_{\ell-1}/\Lambda_{\ell},\sigma_e^2)$ by concatenating $W'(\X_{\ell};\Z|\X_{1:\ell-1})$ and $\widehat{W}$, which means $W(\Lambda_{\ell-1}/\Lambda_{\ell},\sigma_e^2)$ is degraded to $W'(\X_{\ell};\Z|\X_{1:\ell-1})$. Similarly, we can also construct a middle channel $\check{W}$ from $\bar{\Z}$ to $\Z$. For a specific realization $\bar{z}$ of $\bar{\Z}$, this $\widehat{W}$ maps $\bar{z}$ to $\bar{z}+[\Lambda_\ell/\Lambda_r]$ with probability $\frac{1}{|\Lambda_{\ell}/\Lambda_r|}$, where $|\Lambda_{\ell}/\Lambda_r|$ is the order of this partition. This means that $W'(\X_{\ell};\Z|\X_{1:\ell-1})$ is also degraded to $W(\Lambda_{\ell-1}/\Lambda_{\ell},\sigma_e^2)$.

By channel degradation and \cite[Lemma 1]{Ido}, letting channel $W$ and $W'$ denote $W(\Lambda_{\ell-1}/\Lambda_{\ell},\sigma_e^2)$ and $W'(\X_{\ell};\Z|\X_{1:\ell-1})$ for short, we have
\begin{eqnarray}\label{eqn:channeldegradation}
\begin{aligned}
&\widetilde{Z}(W_N^{(i)}) \leq \widetilde{Z}({W'}_N^{(i)}) \text{ and }\widetilde{Z}(W_N^{(i)}) \geq \widetilde{Z}({W'}_N^{(i)}),\\
&I(W_N^{(i)}) \leq I({W'}_N^{(i)}) \;\;\text{ and }I(W_N^{(i)}) \geq I({W'}_N^{(i)}), \notag
\end{aligned}
\end{eqnarray}
meaning that $\widetilde{Z}(W_N^{(i)})=\widetilde{Z}({W'}_N^{(i)})$ and $I(W_N^{(i)}) = I({W'}_N^{(i)})$.
\end{proof}

\bibliographystyle{IEEEtran}
\bibliography{Myreff}
\end{document}